\documentclass[conference,a4paper]{IEEEtran}

\usepackage{amsmath,amssymb}
\usepackage{amsthm}
\usepackage{booktabs}
\usepackage{array}
\usepackage{graphicx}
\usepackage{stfloats}
\usepackage{multirow}
\usepackage{url}
\usepackage{tikz}
\usetikzlibrary{positioning,arrows.meta}
\usepackage{pgfplots}
\pgfplotsset{compat=1.17}
\usepgfplotslibrary{groupplots}
\usepackage{algorithm}
\usepackage{algpseudocode}
\usepackage{xcolor}
\definecolor{FDdark}{RGB}{14,60,107}
\usepackage{hyperref}
\hypersetup{hidelinks}

\theoremstyle{plain}
\newtheorem{theorem}{Theorem}
\newtheorem{lemma}[theorem]{Lemma}
\newtheorem{proposition}[theorem]{Proposition}
\newtheorem{corollary}[theorem]{Corollary}
\newtheorem*{theorem*}{Theorem}
\newtheorem*{lemma*}{Lemma}
\newtheorem*{corollary*}{Corollary}
\theoremstyle{definition}
\newtheorem{definition}[theorem]{Definition}
\newtheorem*{hypothesis}{The two-layer compression hypothesis}
\newtheorem{example}[theorem]{Example}
\theoremstyle{remark}
\newtheorem{remark}[theorem]{Remark}

\newcommand{\taco}{\textsc{taco}}

\newcommand{\Description}[1]{}

\newif\ifarXiv
\ifdefined\DATECUT
\arXivfalse
\else
\arXivtrue
\fi

\newif\ifanonymous
\anonymousfalse
\IEEEoverridecommandlockouts
\ifanonymous
  \author{Anonymous submission}
\else
  \author{\IEEEauthorblockN{Keren Zhu\thanks{Code and data: \url{https://github.com/CODA-Team/TACO}}}
\IEEEauthorblockA{College of Integrated Circuits and Nano-Micro Electronics\\
State Key Laboratory of Integrated Chips and System\\
Fudan University, Shanghai, China\\
krzhu@fudan.edu.cn}}
 \fi

\title{Rethinking Logic Optimization Operators:\\
Theory-Derived Operator Compression\\
via Agentic Source Analysis}

\begin{document}
\pagestyle{plain}
\maketitle

\begin{abstract}
Logic synthesis has evolved from compact two-level minimization to
large multilevel flows with many interacting optimization operators.
Recent work has invested substantial effort in sequencing these
operators: actions are commonly treated as opaque choices in a rapidly
expanding search space, while learned circuit representations and
heuristic or local-greedy orchestration provide increasingly informed
ways to explore it.  A central obstacle is the operator vocabulary
itself.  Production operators are numerous, span different
representations and mathematical foundations, and expose behaviors
determined by implementation-level guards, bounds, and update order.

We address this gap through agentic source analysis, using LLM agents
to formulate operator-level relations from pinned ABC and mockturtle
implementations and adversarial audits to test their stated scope.
The resulting certified relations yield theory-derived operator
compression: 40 deployed recipe actions collapse to a 31-action exact
Pareto cover, and source-level conditions compile into deterministic
admission gates.  We integrate these gates directly into ABC
Orchestrate to form \taco{}.  Two exact gates reduce Orchestrate
runtime by 11\% with bit-identical outputs on 66 circuits.  In a
held-fixed integrated comparison, \taco{} uses fewer nodes on 14 of 16
circuits, with geometric-mean reductions of 1.0\% in nodes and 3.2\%
in levels, while running $2.6\times$ faster.  \taco{}-max achieves an
NDP geometric-mean ratio of $0.903$ on HeLO's three exact-input rows.
\end{abstract}

\section{Introduction}
\label{sec:intro}

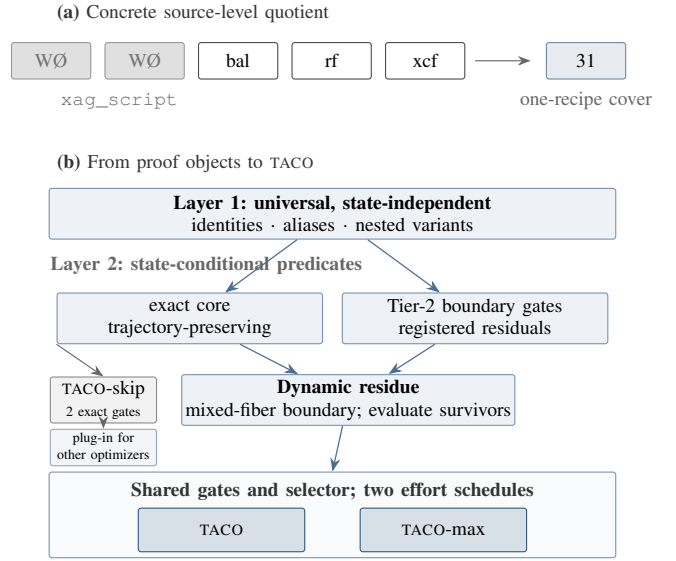
\begin{figure}[!t]
  \centering
  \resizebox{\columnwidth}{!}{\begin{tikzpicture}[
    stage/.style={draw=black!70, rounded corners=1pt, minimum width=1.0cm,
                  minimum height=0.46cm, font=\scriptsize, inner sep=1.3pt},
    ident/.style={stage, fill=black!12, draw=black!40, text=black!55},
    lbl/.style={font=\scriptsize, text=black!60},
    main/.style={stage, minimum width=1.55cm, minimum height=0.64cm,
                 align=center, fill=FDdark!7, draw=FDdark!55},
    taco/.style={main, fill=FDdark!16, draw=FDdark!85},
    mode/.style={stage, minimum width=2.10cm, minimum height=0.42cm,
                 fill=FDdark!11, draw=FDdark!65, align=center},
    modestatic/.style={mode, fill=FDdark!16, draw=FDdark!85},
    arrow/.style={-{Stealth[length=1.6mm]}, black!60},
    mainarrow/.style={-{Stealth[length=1.7mm]}, FDdark!75, line width=0.45pt}
  ]
\node[lbl, anchor=west, text=black!80] at (-0.05,1.05)
    {\textbf{(a)} Concrete source-level quotient};
  \node[ident] (s1) at (0,0.45) {W\O{}};
  \node[ident] (s2) at (1.18,0.45) {W\O{}};
  \node[stage] (s3) at (2.36,0.45) {bal};
  \node[stage] (s4) at (3.54,0.45) {rf};
  \node[stage] (s5) at (4.72,0.45) {xcf};
  \node[lbl, anchor=west] at (0,-0.05) {\texttt{xag\_script}};
  \draw[arrow] (5.35,0.45) -- (6.05,0.45);
  \node[stage, fill=FDdark!10, draw=FDdark!60] (c1) at (6.75,0.45) {31};
  \node[lbl] at (6.75,-0.05) {one-recipe cover};

\node[lbl, anchor=west, text=black!80] at (-0.05,-0.85)
    {\textbf{(b)} From proof objects to \taco{}};
  \node[main, minimum width=7.15cm] (u) at (3.55,-1.48)
    {\textbf{Layer 1: universal, state-independent}\\
     identities $\cdot$ aliases $\cdot$ nested variants};
  \node[lbl, anchor=west, fill=white, inner xsep=1pt] at (-0.05,-2.13)
    {\textbf{Layer 2: state-conditional predicates}};
  \node[main, minimum width=3.35cm] (exact) at (1.75,-2.80)
    {exact core\\trajectory-preserving};
  \node[main, minimum width=3.35cm] (tier2) at (5.35,-2.80)
    {Tier-2 boundary gates\\registered residuals};
  \node[stage, fill=black!5, draw=black!45, minimum width=1.34cm,
        minimum height=0.58cm, align=center] (skip) at (0.66,-3.83)
    {\taco{}-skip\\{\tiny 2 exact gates}};
  \node[stage, fill=FDdark!5, draw=FDdark!45, minimum width=1.34cm,
        minimum height=0.44cm, font=\tiny, align=center] (skipout)
        at (0.66,-4.43)
    {plug-in for\\other optimizers};
  \node[main, minimum width=4.15cm] (residue) at (3.75,-3.83)
    {\textbf{Dynamic residue}\\mixed-fiber boundary; evaluate survivors};
  \draw[mainarrow] (u) -- (exact);
  \draw[mainarrow] (u) -- (tier2);
  \draw[arrow] (exact.south west) -- (skip.north);
  \draw[arrow] (skip) -- (skipout);
  \draw[mainarrow] (exact) -- (residue);
  \draw[mainarrow] (tier2) -- (residue);

\node[draw=FDdark!55, rounded corners=1pt, minimum width=7.15cm,
        minimum height=1.08cm, fill=FDdark!2] (modes) at (3.55,-5.28) {};
  \node[lbl, text=black!80] at (3.55,-4.98)
    {\textbf{Shared gates and selector; two effort schedules}};
  \node[modestatic, minimum height=0.55cm] at (2.15,-5.47)
    {\taco{}};
  \node[modestatic, minimum height=0.55cm] at (4.95,-5.47)
    {\taco{}-max};
  \draw[mainarrow] (residue) -- (modes.north);
\end{tikzpicture}
\Description{Panel (a) shows a five-stage inherited recipe whose first
two calls are identities and a 31-action one-recipe cover.  Panel (b)
shows a state-independent universal layer, a state-conditional layer
split into an exact core and Tier-2 boundary gates, the dynamically
evaluated residue, a separate gate-only diagnostic, and two
deterministic TACO executions with different effort schedules.}
 }
  \caption{A concrete quotient and the layered system.  (a) In
  \texttt{xag\_script}, W\O{} is a cut-rewriting call whose returned
  network is discarded, so it is an identity at recipe scope; across
  the recipe universe,
  the quotient leaves a 31-action one-recipe Pareto cover.  (b)
  Universal and state-conditional layers delimit the dynamic residue.
  \taco{} and \taco{}-max share one deterministic gate layer and
  selector; maximum effort uses an iterative AIG pre-pass and larger
  MIG-round and time limits.
  \taco{}-skip packages two exact gates as a plug-in for other
  optimizers; on \texttt{orchestrate}, it is 66/66 identical and
  11\% faster.}
  \Description{Panel (a) shows a recipe with two discarded-return
  identity stages and the 31-action cover.  Panel (b) shows universal
  relations, exact and Tier-2 boundary gates, a dynamic residue,
  the gate-only diagnostic and its measured outcome, and the two
  deterministic TACO effort schedules.}
  \label{fig:twist}
\end{figure}

Optimization concerns both the transformations a system can perform
and the possibilities worth considering.  As logic
synthesis has accumulated representations, operators, and scales of
decision, that second question has become increasingly central.

Logic optimization did not begin as an operator-selection problem.
In classical two-level synthesis, the representation and its
minimization problem were closely aligned: algorithms optimized an
explicit Boolean cover within a shared form~\cite{brayton1984espresso}.
The move to multilevel networks replaced that compact decision object
with factored DAGs whose quality depends on interacting local
transformations~\cite{brayton1987mis,sentovich1992sis}.  DAG-aware
rewriting, don't-care-based resubstitution, exact synthesis,
reconvergence-driven windows, and simulation enlarged what each local
step could see or attempt~\cite{mishchenko2006rewriting,
mishchenko2011resubstitution,haaswijk2020exact,riener2022window,
lee2022simulation}.  MIG, XAG, and XMG added carrier-specific algebra
and cost models~\cite{amaru2016mig,testa2020xag,soeken2019xmg}.
Each addition answered a real structural or functional need, but no
representation, operator, or ordering is uniformly best.  The richer
toolbox therefore leaves a residual problem: choosing among its tools.

Once several passes are available, their ordering creates a second
problem.  DRiLLS, FlowTune, BOiLS, AlphaSyn, and related systems
therefore cast synthesis as sequence or parameter
search~\cite{hosny2020drills,launeto2022flowtune,grosnit2022boils,
pei2023alphasyn,zhu2020exploring}.  Retrieval, learned circuit
representations, and surrogate models improve decision quality or sample
efficiency~\cite{chowdhury2024retrieval,wu2022lostin,wu2023gamora,
li2024boolgebra}.  These methods address genuine order sensitivity,
with the available action vocabulary naturally serving as their
starting point.

Multiple DAG carriers add another axis of choice.  LSOracle exposes
carrier-specific optimization, while HeLO responds to heterogeneous
modules and the cost of evaluating every carrier optimizer by
partitioning the circuit, predicting one best-fit DAG type per
subcircuit, and invoking its fixed LSOracle optimization
script~\cite{neto2019lsoracle,pu2025helo}.  The script is a useful
coarse-grained abstraction: it packages low-level operator orderings
while learning focuses on representation and partition choice.

DAG-aware orchestration moves the decision to a finer granularity:
a stand-alone pass can overlook roots better served by another
engine~\cite{li2024orchestration}.  It moves selection inside the
graph, but its local-greedy policy evaluates rewrite, resubstitution,
and refactoring at each root before choosing the largest local gain;
its cheaper alternative follows a fixed priority order.  The original
paper accordingly identifies learned, domain-aware selection as a
future direction.  Together, these works move the decision boundary
from flow sequences to carrier scripts on subcircuits and finally to
individual roots.

This progression leaves a complementary question: \emph{before
selection, can source-level relations identify choices that are
equivalent or evaluations that are unnecessary?}  Mature open tools contain
identities, aliases, and nested variants accumulated over long
development histories~\cite{brayton2010abc,soeken2018mockturtle}.
Figure~\ref{fig:twist}(a) shows the smallest visible example: a pure
cut-rewriting call whose returned network is discarded cannot change
the recipe state.  These relations characterize the action menu as a
whole.

A natural first route is classical: abstract each transformation as an
oracle, seek a circuit class closed under the action set, and prove one
reusable theorem for that class.  A representative construction
\ifarXiv in Appendix~\ref{app:closedset}\fi\ illustrates why practical
conclusions are difficult to derive at this level: without
\emph{source semantics}, including guard constants, view limits,
accounting rules, and update order, mathematical closure alone does not identify
which deployed choices are equivalent or unnecessary.  Artifact-specific
verification has a mature compiler precedent: Alive and Alive2 validate
concrete LLVM rewrites against implementation and language
semantics~\cite{lopes2015alive,lopes2021alive2}.  We apply the same
discipline to relations among implemented synthesis operators.

We therefore analyze concrete, pinned operators and prove
operator-specific lemmas at implementation scale.  Language-model
systems have demonstrated capabilities in mathematical discovery and
theorem solving~\cite{romeraparedes2024funsearch,trinh2024alphageometry},
including revisiting classical algorithmic
questions~\cite{zhu2026primdijkstra}.  \textbf{Agentic source analysis
makes this mathematical-engineering workload tractable by assigning
lemma-scale source readings to LLM agents and subjecting retained
relations to adversarial audit.}

Figure~\ref{fig:twist}(b) summarizes this two-layer decomposition.
Universal relations remove state-independent redundancy,
state-conditional gates avoid idle evaluations, and a scoped fiber
obstruction marks part of the dynamic residue.  Their reusable proof
objects are fixed-scope quotients, sound admission predicates, and
witnessed mixed-fiber obstructions.  Section~\ref{sec:universal}
formalizes this \emph{two-layer compression hypothesis} for pinned ABC
and mockturtle.

The registry separates an exact, trajectory-preserving core from
Tier-2 boundary gates with registered residuals.  \taco{} inserts this
fixed gate set directly into ABC's pinned
\texttt{Abc\_NtkOrchLocal}, while preserving its AIG representation,
root traversal, stock engines, source parameters, and published mode
schedules.  A deterministic selector resolves the surviving residue.
\taco{}-max retains the gate set and selector but uses the
higher-effort AIG/MIG schedule for the HeLO comparison; \taco{}-skip
exposes the exact gate-only subset as a plug-in.

Our contributions are summarized as follows.
\begin{itemize}
\item We propose agentic source analysis as a paradigm for deriving
theory from implemented optimizers, coupling lemma-scale readings of
pinned source with adversarial audit and a replayable claim registry.

\item We formulate the two-layer compression hypothesis through three
proof objects: a 40-to-31 fixed-scope quotient, sound admission
predicates for per-root gates, and witnessed mixed-fiber obstructions
for the remaining dynamic residue.

\item We develop \taco{}, a deterministic gate-conditioned optimizer
inside ABC Orchestrate that preserves its published AIG execution
skeleton and evaluates only choices left unresolved by the gates.

\item We evaluate the exact gates and integrated system separately:
the gates cut runtime by 11\% with bit-identical outputs on 66
circuits, \taco{} uses fewer nodes on 14 of 16 Orchestrate circuits at
$2.6\times$ speed, and \taco{}-max reaches $0.903$ NDP against HeLO's
exact-input rows.
\end{itemize}

The rest of the paper is organized as follows.
Section~\ref{sec:universal} fixes the setting and proves the universal
layer, Section~\ref{sec:predicate} the predicate layer, and
Section~\ref{sec:crosscarrier} the cross-carrier arithmetic and the
impossibility results.  Section~\ref{sec:optimizer} builds the
optimizer, Section~\ref{sec:production} describes the production
method, and Section~\ref{sec:experiments} reports the
experiments.\ifarXiv\
Section~\ref{sec:discussion} discusses limits and open questions.\fi\
Section~\ref{sec:conclusion} concludes.
 \section{Fixed-Scope Quotients and Universal Compression}
\label{sec:universal}

The compressions in this section are unconditional.  When two actions
can be shown to reach the same endpoint on every state, one of them can
be deleted from every flow, and the action set
shrinks before any search begins.
The menu was assembled by accretion: recipes
inherited stages, libraries grew parameter variants, and the duplicates
were never accounted for.  We find three kinds of state-independent
redundancy:
\begin{itemize}
\item \emph{identities}: actions that cannot change the network;
\item \emph{aliases}: two actions that are the same operator;
\item \emph{nested variants}: a variant whose possible endpoints are
already reached by a small basis.
\end{itemize}
Each relation is established from the pinned implementation semantics
at its stated scope; its source anchors and proof obligation are
recorded in the released registry.

\subsection{The flat action space}
\label{sec:setting}

A \emph{state} is a finite logic network on one of the carriers of the
multilevel tradition~\cite{brayton1987mis,sentovich1992sis}: AIG
(and-inverter graph~\cite{kuehlmann2002dag}), MIG
(majority-inverter graph~\cite{amaru2016mig}),
XAG (XOR-AND graph~\cite{testa2020xag}), or XMG (XOR-majority
graph~\cite{soeken2019xmg}).  Throughout, \emph{pinned} refers to a fixed
source tree for each library: ABC at commit \texttt{bcfdf59}, mockturtle at
\texttt{fb8f879}, and LSOracle at \texttt{febd152} for the recipe
framework; the artifact builds exactly these trees.  The action set
under study has four parts.
(1) The \emph{recipe universe}: the 40 script
actions of HeLO's framework~\cite{pu2025helo}, inherited from
LSOracle~\cite{neto2019lsoracle} and numbered H00--H39
\ifarXiv\ (Table~\ref{tab:universe})\fi.  (2) The \emph{AIG engines} of ABC:
rewrite, resubstitution, refactoring, and the \texttt{orchestrate}
command~\cite{li2024orchestration} built on them, which is the
algorithm our gates accelerate.  (3) The \emph{MIG algorithms}: the
mockturtle passes used by our MIG-side flow.  (4) Our optimizer
\taco{}, built from the compressed set and compared experimentally
with ABC Orchestrate and HeLO's published results.

We separate implementation semantics from corpus measurement.  Let
$S_{\mathrm{impl}}$ be the legal states accepted by the pinned
implementations.  For a named benchmark set $D$ and action set $A$, let
$R(D,A)\subseteq S_{\mathrm{impl}}$ be the states encountered from $D$
under the stated flow.  Source-level theorems quantify over
$S_{\mathrm{impl}}$ (or an explicitly stated subset); measurements
quantify only over the named $D$ or $R(D,A)$.

We also distinguish three semantic scopes.  A \emph{root evaluator}
$E_e(s,v)$ commits at most one candidate of engine $e$ at root $v$ (or
returns $s$); a \emph{pass} $P_e(s)$ applies the pinned visit and commit
schedule over a network; and a recipe action $T_h(s)$ composes complete
stages.  A candidate is a replacement subgraph with integer gain $g$,
the size of the root's fanout-free cone minus the reported insertion
count, and resulting root level $\ell$.  Relations below compare
transformations at the same scope: the 40-to-31 quotient is recipe-level,
the selector lemmas are root-level, and each predicate names the exact
evaluation it gates.  Whole-pass claims require a separate relation.
We write $I$ for the identity at any scope.
\emph{Bridge actions}
change the carrier and are first-class actions with their own costs; a
\emph{lane} (bridge in, operators, bridge out) is treated as a macro
and lies outside the formal relations.  Table~\ref{tab:notation} collects
the notation used throughout, and Table~\ref{tab:space} the compressed
action set this section produces.

\begin{table}[t]
\centering
\scriptsize
\caption{The compressed action set: the retained menu (top), and
removals marked \dag~(removed permanently by theorem),
\ddag~(skipped per state by a proved predicate), and \S~(a
source-derived or measured Tier-2 contract).}
\label{tab:space}
\begin{tabular}{@{}l >{\raggedright\arraybackslash}p{0.56\columnwidth} @{}}
\toprule
\multicolumn{2}{@{}l}{\emph{retained menu}} \\
\midrule
AIG & Rewrite (4-cut); resubstitution ($\le$2 inserted nodes); \\
    & Refactoring (source-default 10/16-node cone limits); balancing \\
MIG & Algebraic depth rewriting (3 strategies); cut rewriting \\
    & (4-cut); resubstitution; SOP balancing; Akers refactoring \\
Lanes & Whole-network excursions into MIG, AIG, XAG, XMG \\
\midrule
\multicolumn{2}{@{}l}{\emph{removed or gated}} \\
\midrule
Strict-gain-only variants \dag & Zero-gain variant covers them \\
0/1-Insert resubstitution \dag & Nested inside the 2-insert version \\
4-Leaf refactoring cones \dag & Same operator under an applicability guard \\
No-area algebraic \dag & Contained in the area-allowed variant \\
3-Leaf cut rewriting \dag & Dominated by 4-leaf \\
Critical-path balancing \dag & Contained in all-node balancing \\
Saturation flag (FR) \dag & Dead guard; never fires \\
XAG Constant-fanin opt. \ddag & Identity on XOR-free states
 (G-XAGCF) \\
XMG boundary skip \S & Registered Tier-2 gate; collision mechanism in
 Thm.~\ref{thm:xmg} \\
SOP passes on XOR-rich cones \S & ISOP covers explode
 (aes\_core $>$ 20\,min per pass) \\
\bottomrule
\end{tabular}
\end{table}

\begin{definition}[containment, alias, redundancy]
\label{def:containment}
At one fixed semantic scope, for a transformation $F$ and a set of
transformations $\mathcal{G}$ with the same input domain, write
$F\preceq\mathcal{G}$ when
$F(x)\in\{I(x)\}\cup\{G(x):G\in\mathcal{G}\}$ for every legal input
$x$, and call $F$ redundant against $\mathcal{G}$.  Write $F\equiv G$
when their endpoints agree pointwise, and call them aliases.  Every use
below states whether $x$ is a state $s$ or a state--root pair $(s,v)$.
The relations are for the implementations as compiled; audit anchors
are collected in the released claim manifest
\ifarXiv and Appendix~\ref{app:proofsuniv}\fi.
\end{definition}

\begin{definition}[exact Pareto cover]
\label{def:paretocover}
For a recipe set $A$ and input $C$, let
$\mathrm{PF}_{A}(C)$ be the Pareto frontier of the one-recipe endpoints
$\{q(T_h(C)):h\in A\}$ under the reported node--depth metric $q$.
A subset $B\subseteq A$ is an \emph{exact Pareto cover} on an input set
$D$ when $\mathrm{PF}_{B}(C)=\mathrm{PF}_{A}(C)$ for every $C\in D$.
``Exact'' denotes equality of these one-action frontiers.  Minimum
cardinality and arbitrary multi-step trajectories lie outside this
definition.
\end{definition}

\begin{table}[t]
\centering
\scriptsize
\caption{Notation and named concepts used throughout.}
\label{tab:notation}
\begin{tabular}{@{}p{0.28\columnwidth}p{0.66\columnwidth}@{}}
\toprule
Term / symbol & Meaning \\
\midrule
State $s$ & A finite network in $S_{\mathrm{impl}}$ \\
$E_e(s,v)$; $P_e(s)$; $T_h(s)$ & Root evaluation; pass; recipe endpoint \\
$I(s)$ & The identity endpoint (no change) \\
Strash & Structural hashing, the canonical AIG construction \\
Candidate & A proposed replacement, with gain $g$ and level $\ell$ \\
Gain $g$ & Cone size minus inserted nodes, at a root \\
$F \preceq \mathcal{G}$ & Same-scope endpoint is $I$ or one endpoint in $\mathcal{G}$ \\
$a \equiv b$ & Equal endpoints on every state (alias) \\
Cut ($k$-leaf) & A small cone under a node, at most $k$ leaves \\
MFFC & Maximal fanout-free cone rooted at a node \\
NPN Class & Functions equal up to negation and permutation \\
NDP & Node-depth product, the HeLO objective \\
R2, R2p & Resubstitution, 2 inserted nodes; depth-preserving variant \\
Record & An allocated object in the node manager, live or dead \\
Dry & An engine or pass with no win over a measurement block \\
Required level & Arrival-time bound; \emph{critical} = zero slack \\
Offer / win & Returns a positive-gain candidate / the selector commits it \\
Identity region & States where the operator provably changes nothing \\
Pareto frontier & Mutually undominated (nodes, depth) outcomes \\
Recipe universe & The 40 inherited script actions (H00--H39) \\
Pinned & The fixed source trees of \S\ref{sec:setting} \\
Bridge action / lane & A carrier-changing action / bridge plus operators \\
Stored size $s^*$ & Gate count of an NPN class's stored database form \\
Trajectory & The sequence of states a flow visits \\
MAJ-compressible & MIG form strictly smaller than AIG form \\
\bottomrule
\end{tabular}
\end{table}

\begin{hypothesis}
\label{sec:hypothesis}
Let $\mathcal{A}$ be the operator set of a production logic-synthesis
tool, acting on its legal implementation states.  We hypothesize that
its operator relations admit the following decomposition:
\begin{enumerate}
\item[(i)] there is a basis $B \subseteq \mathcal{A}$ whose alias and
metric-dominance relations are certifiable at an explicitly named
semantic scope, preserving the corresponding one-action Pareto
frontier;
\item[(ii)] the idle behavior of an operator $e \in B$ is partly
certifiable by sound structural gates: sterile or dedupe gates
(Definition~\ref{def:gates}) that skip provably pointless evaluations
without claiming to find them all;
\item[(iii)] when an operator's identity status varies within one
polarity-erased structural class, no predicate over that structure can
be both sound (fires only on identities) and complete (fires on every
identity); this delimits a dynamic residue.
\end{enumerate}
\end{hypothesis}

\begin{definition}[function-dependent operator]
\label{def:fdep}
An operator is \textbf{function-dependent} when its acceptance rule reads
Boolean function data (truth tables, canonical forms, simulation)
beyond the polarity-erased structure.
\end{definition}

\noindent
The three clauses form an analysis framework.  For the pinned ABC and
mockturtle trees, each relation and gate is proved at its stated scope.
The recipe quotient preserves the one-action frontier over every legal
input to the pinned recipe implementation, while corpus enumeration
independently rechecks the quotient and records which representatives
are live.  A concrete polarity pair instantiates the structural
obstruction for cut-4 rewriting.
Clause (i) is the subject of this section, clause (ii) of
Section~\ref{sec:predicate}, and clause (iii) of
Section~\ref{sec:crosscarrier}.

\subsection{The discarded-cut lemma and quotient chain}

Some recipe stages cannot change anything because the operator they
call is a pure function and the recipe discards its result.

\begin{lemma}[discarded cut]
\label{lem:discarded}
Let $\rho$ be any recipe stage that invokes cut rewriting on a state
$s$ without binding the returned network.  Then $T_\rho(s) = s$ for
every state $s$.
\end{lemma}

\noindent
An entire family of parameterized recipes then collapses.

\begin{corollary}[quotient and exact Pareto cover]
\label{cor:quotient}
Eleven of the 40 recipe actions collapse into three alias classes
(four AIG-script variants that become identical once their discarded
stages are removed, five MIG-script variants of the same kind, and two
XAG-script variants), removing eight actions.  One further action is
metric-dominated (the no-op recipe action H00).  The remaining 31
representatives form
an exact Pareto cover in the sense of
Definition~\ref{def:paretocover} on the complete legal input domain of
the pinned recipe implementation.  Pointwise aliases and the universal
$H27$--$H00$ dominance preserve the
one-recipe frontier for every legal input.  Corpus enumeration provides
an independent regression and liveness check.  Minimum cardinality
remains open.
\end{corollary}

\begin{proof}[Proof skeleton]
Substituting Lemma~\ref{lem:discarded} for every unbound cut-rewrite
stage makes four AIG recipes equivalent, five MIG recipes equivalent,
and two XAG recipes equivalent.  Their remaining parameter differences either feed an
identity stage or the functional-reduction saturation flag.  That flag
reads an append-only record count; the pass allocates no new record, so
the flagged and default transitions coincide at every iteration.
These three alias classes remove eight actions.

For the remaining dominance, H00 is the identity recipe.  H27 imports
the same graph into the MIG carrier, applies one no-area algebraic
depth pass, and lowers it back.  Each admitted move replaces two
majority nodes by at most two and strictly reduces level; cleanup
cannot worsen either coordinate.  Hence
$q_{H27}(C)\le q_{H00}(C)$ for every legal input $C$.  Removing this
dominated action leaves 31 representatives, and pointwise equality or
coordinatewise dominance preserves the one-recipe Pareto frontier.
\end{proof}

Certificate \texttt{U-Q31} records this result at vocabulary scope:
every
deployed action has an explicit representative or dominator at the
same semantic scope, and the retained vocabulary preserves the full
one-recipe frontier.  Discarded returns, inert guards, and storage
rules supply local premises for that scoped composition.

\begin{proposition}[two non-collapsible look-alike pairs]
\label{prop:sharp}
Two tempting additional merges are invalid: the MIG and XMG algebraic
scripts are not the same action, and
the no-area MIG variant is not dominated by its XMG twin.  A
seven-AND single-output AIG witnesses the first separation and a
nine-gate witness the second, both through the old-node hash asymmetry
of \S\ref{sec:xmg}.
\ifarXiv Appendix~\ref{app:refutations} records the complete refutation
trail.\fi
\end{proposition}

\subsection{Selector lemmas}

The next five relations concern nested variants.  They cover every
nested variant pair in the pinned action set: each endpoint reachable by the smaller
variant is also reachable by the larger.  Formally, each relation is a
$\preceq$ statement between engine variants and follows from their
shared acceptance guard.

\begin{lemma}[per-root selector compressions]
\label{lem:selectors}
Let $a$ and $a'$ be two variants of one engine differing in a single
parameter.  Then the following hold for every legal state--root pair
$(s,v)$, with each endpoint interpreted as $E_a(s,v)$:
\begin{enumerate}
\setlength{\itemsep}{1pt}
\setlength{\parskip}{0pt}
\setlength{\parsep}{0pt}
\item[(a)] strict-gain-only $a$:
$E_a(s,v) \in \{E_{a'}(s,v), s\}$;
\item[(b)] no-area algebraic $a$:
$E_a(s,v) \in \{E_{a'}(s,v), s\}$;
\item[(c)] critical-only balancing $a$:
$E_a(s,v) \in \{E_{a'}(s,v), s\}$;
\item[(d)] resubstitution with insertion budget $k \in \{0,1\}$:
$E_{R_k}(s,v)\in\{E_{R_{k+1}}(s,v),s\}$; the larger budget reaches
every first success of the smaller;
\item[(e)] a 4-leaf refactoring limit:
$E_{D_4}(s,v)\in\{E_{D_6}(s,v),s\}$, because the default 6-leaf
variant evaluates every admitted cone identically.
\end{enumerate}
\end{lemma}

These relations shrink the candidate basis of a \emph{single root
evaluation}.  Whole-pass endpoints have a separate scope because later
visits can observe an earlier commit.  Likewise, the 3-leaf/4-leaf row
of Table~\ref{tab:space} states recipe-level metric dominance under its
reported objective.

\subsection{Gain-form theorems}

The last two lemmas provide gain contracts for the selector.  A
per-root selection rule compares reported gains, so each
engine needs a proved relation between that report and the realized
node-count change.  The next two lemmas pin the reported gain ranges;
for resubstitution the report is also a lower bound on realized
reduction.  The resubstitution lemma also rules out the selector's
zero-gain branch.

\begin{lemma}[RS gain]
\label{lem:rsgain}
Let $c$ be a resubstitution candidate at any root $v$ of any state $s$.
Then $g(c) \ge 1$.
\end{lemma}

\begin{proof}[Proof sketch]
Candidates are tried in increasing insertion budget with stop-gates on
the cone size: a zero-insertion success yields $g = |\mathrm{MFFC}(v)|
\ge 1$, and a $k$-insertion candidate is reached only when the cone has
at least $k{+}1$ records, giving $g \ge 1$ again.
\end{proof}

\begin{lemma}[RW/RF gain]
\label{lem:rwgain}
Let $c$ be a rewrite or refactoring candidate at any root $v$.  Then
$0 \le g(c) \le |\mathrm{MFFC}(v)|$.
\end{lemma}

\noindent
Table~\ref{tab:universal} collects the layer.

\begin{table}[t]
\centering
\scriptsize
\caption{The universal layer: source-certified relations at their
stated semantic scopes.}
\label{tab:universal}
\begin{tabular}{@{}ll@{}}
\toprule
Statement & Mechanism \\
\midrule
Unbound cut-rewrite stages $= I$ & Pure function \\
Saturation variants $\equiv$ default & Record counting \\
3 alias classes $+$ one dominance & Lem.~\ref{lem:discarded} \\
31-action exact Pareto cover & Recipe quotient; metric dominance \\
Lem.~\ref{lem:selectors} (a)--(e) & Guard reading \\
RS gain $\ge 1$; RW/RF gain $\in [0,\,|\mathrm{MFFC}|]$ & Stop-gates; accounting \\
\bottomrule
\end{tabular}
\end{table}
 \section{Sound Admission Predicates}
\label{sec:predicate}

The predicate layer partitions the state space and gates each region
separately: a cheap structural check on the current state decides
whether an operator's evaluation there is provably pointless, and only
those evaluations are skipped.  Universal relations remove an operator
over the entire implementation domain; predicate gates make the finer
decision separately for each state or root.  The
operators that survive still have idle regions: rewriting has nothing
to rewrite at a shallow root
(Theorem~\ref{thm:no4cut}), and algebraic depth rewriting has nothing to
reassociate on a balanced network
(Theorem~\ref{thm:gap}).  Gates come in three kinds, sterile,
dedupe, and accounting (Definition~\ref{def:gates}).
A \emph{gate} below is always such a predicate, never a network node.
Every \emph{trajectory-certified} skip leaves the trajectory (the
sequence of states the flow visits) bit-identical.  The first three
results below license such skips, and the fourth provides exact
accounting without skipping.  The final case proves the only mechanism by
which two carrier implementations can diverge; the cheaper
boundary-only test derived from it remains a source-derived Tier-2
choice with a separately measured residual.

\begin{definition}[structural predicate]
\label{def:structural}
Let $\sigma(s)$ be the directed typed multigraph of $s$ after erasing
every edge-complement bit and every primary-input variable name, while
retaining adjacency, multiplicity, node type (constant, primary input,
or gate operator), outputs, and the distinguished root when one is
specified.  A \emph{structural predicate} is any predicate that factors
through $\sigma$: $\pi(s)=\widehat{\pi}(\sigma(s))$, or
$\pi(s,v)=\widehat{\pi}(\sigma(s,v))$ at root scope.  It may therefore
use the entire polarity-erased graph, including information beyond
aggregate counts, but
it cannot use truth tables, simulation values, input identities, or
edge polarity.
\end{definition}

\begin{definition}[sterile, dedupe, and accounting gates]
\label{def:gates}
A structural predicate $\pi$ is a \textbf{sterile} gate for a root
evaluator $E_e$ if $\pi(s,v)=1$ implies $E_e(s,v)=s$, and for a pass
$P_e$ if $\pi(s)=1$ implies $P_e(s)=s$.  It is a \textbf{dedupe} gate
for two transformations at the same scope if $\pi=1$ implies equal
endpoints at that scope.  An \textbf{accounting} gate only bounds the
gain an evaluator can return; it skips nothing.  Every gate below names
its scope.  Only a proved sterile or dedupe implication licenses a
trajectory-equality claim; empirical gates are labeled separately.
\end{definition}

\subsection{The gate shortlist}

Resubstitution admits divisors only when they satisfy margins derived
from the required level, so a depth-preserving variant that clamps the required
level to the old root level can differ from the free variant only where
the root has slack.  On a critical root there is no slack, and the two
variants produce the same endpoint.

\begin{theorem}[dedupe on critical roots]
\label{thm:r2p}
Let $R$ be free resubstitution and $R^{*}$ its depth-preserving variant
with required level clamped to $\ell(v)$.  Applied at a critical root
$v$, the two root evaluators coincide:
$E_{R^{*}}(s,v) = E_R(s,v)$.
\end{theorem}

Algebraic depth rewriting accepts a reassociation at a node only when
the node's two deepest children differ in level by more than one.  With
no such node anywhere, all three strategies (selective, depth-first,
and aggressive) have nothing to accept.

\begin{theorem}[level-gap guard, sterile]
\label{thm:gap}
If no node $n$ of $s$ satisfies
$\ell(\mathrm{deep}(n)) > \ell(\mathrm{second}(n)) + 1$, where
$\mathrm{deep}(n)$ and $\mathrm{second}(n)$ are $n$'s deepest and
second-deepest children, then each of the three algebraic strategies
acts as $I$ on $s$.
\end{theorem}

Rewrite's evaluation loop considers only cuts with exactly four leaves; smaller
cuts are enumerated and skipped.  Shallow roots with a
primary-input fanin cannot supply four leaves.

\begin{theorem}[no-4-cut, sterile]
\label{thm:no4cut}
If $\ell(v) = 2$ and one fanin of $v$ is a primary input, then every
cut of $v$ has at most three leaves, and the rewrite engine proposes no
candidate at $v$.
\end{theorem}

\begin{proof}
A primary input contributes only its trivial cut; a level-1 fanin
contributes at most its two level-0 fanins; a cut of $v$ is a union of
the two, of size at most $1 + 2 = 3$.
\end{proof}

The gate-only diagnostic pairs G-RW with a second exact stock-path
filter, G-RF.  At a genuine two-input AND of distinct nonconstant
primary inputs, refactoring reconstructs the root itself, and the
identity-hash guard rejects that candidate.  G-RF is therefore a
sterile gate for stock \texttt{orchestrate}; integrated \taco{} uses its
broader registered RF predicates instead.

The cut-4 accounting result characterizes when rewriting can yield
positive gain.  Only 27 of the 222 classes have stored size
at most 3, which explains why firing at low MFFC sizes is confined to
a small subset of classes; structural-hash
reuse can enlarge that region.  The accounting is exact relative to
the implementation because every replacement comes from its class's
stored MIG form.  Positive gain requires the removed cone to exceed
the unreferenced replacement cost, which is bounded by the stored size.

\begin{theorem}[cut4 accounting]
\label{thm:cut4}
Let $c$ be a 4-leaf cut at node $n$ whose NPN class has \emph{stored
size} $s^*$ (the gate count of the class's stored MIG database form).
The pass fires at $n$ through $c$ only if the replacement's
unreferenced gate count is below $|\mathrm{MFFC}(n)|$; the count is at
most $s^*$.  Hence (i) the pass offers nothing on one-record cones, and
(ii) any node with $|\mathrm{MFFC}(n)| \ge 8$ and a genuine cut (at
least 3 leaves) admits an offer, since the 4-variable table has
$\max s^* = 7$.
\end{theorem}

\subsection{Case study: the XMG gate}
\label{sec:xmg}

The gates so far concern single engines.  The final one
concerns two carriers at once; we write $\mathrm{and}(x,y)$ for the
two-input AND, encoded $\mathrm{maj}(0,x,y)$ in a majority carrier.
The MIG and XMG algebraic engines apply
the same majority rules; they differ in exactly one detail of
substitution: the MIG engine rejects a substitution whose rebuilt
parent coincides with the node being replaced, and the XMG engine does
not.  The engines can therefore diverge only through a self-referential
substitution, which requires the shape $p = \mathrm{and}(n,x)$ with
$n = \mathrm{and}(x,y)$: replacing $n$ by $y$ rebuilds a parent with
exactly $n$'s signature (Fig.~\ref{fig:xmg}).

\begin{theorem}[XMG--MIG collision necessity]
\label{thm:xmg}
Run the XMG and MIG majority components from the same state, with the
same visit order.  If the two executions first diverge, then immediately
before that step the common state contains an adjacent node--parent pair
that shares a fanin.  Consequently, if no such pair occurs at any
intermediate state of their common execution trace, the two pass
endpoints coincide.
\end{theorem}

The theorem identifies a necessary immediate source-level condition
for divergence.  A complete lane schedule lies beyond its scope
because reassociation can manufacture the shape internally.  The
deployed G-P6 shortcut rechecks for the pair after the XMG DFS stage and
treats absence as a Tier-2 boundary gate in \taco{}.  Its per-gate
ablation changes s38417 from NDP 136{,}918
to 137{,}173 and leaves the other five $H_6$ cases unchanged.  The G-P6
registry entry therefore carries a measured Tier-2 contract.  An
unconditional MIG--XMG equivalence is false.
\ifarXiv Appendix~\ref{app:xmg-correction} records the correction
trail.\fi

\begin{figure}[t]
  \centering
  \resizebox{\columnwidth}{!}{\begin{tikzpicture}[
    gate/.style={circle, draw=black!70, inner sep=0pt, minimum size=13pt, font=\footnotesize},
    lbl/.style={font=\footnotesize, text=black!60},
    panel/.style={font=\footnotesize, text=black!75, align=center}
  ]
\node[gate] (n) at (0,1.1) {$n$};
  \node[gate] (p) at (1.3,0) {$p$};
  \node[lbl] at (-0.75,1.9) {$x$};
  \node[lbl] at (0.45,1.9) {$y$};
  \draw (n) -- (p);
  \draw (-0.45,1.7) -- (n);
  \draw (0.45,1.7) -- (n);
  \draw (-0.45,1.7) .. controls (1.9,1.3) .. (p);
  \node[panel, anchor=north] at (0.65,-0.8) {(a) $p=\mathrm{and}(n,x)$,\\ $n=\mathrm{and}(x,y)$};
\begin{scope}[shift={(4.9,0)}]
    \node[gate, draw=FDdark, line width=0.9pt] (n2) at (0,1.1) {$n'$};
    \node[gate] (p2) at (1.3,0) {$p'$};
    \node[lbl] at (-0.75,1.9) {$x$};
    \node[lbl] at (0.45,1.9) {$y$};
    \draw (n2) -- (p2);
    \draw (-0.45,1.7) -- (n2);
    \draw (0.45,1.7) -- (n2);
    \draw (-0.45,1.7) .. controls (1.9,1.3) .. (p2);
    \draw[FDdark, line width=0.9pt, -{Stealth[length=1.8mm]}]
      (n2) .. controls (0.9,0.75) .. node[lbl, pos=0.45, above right=-1pt, text=FDdark] {same signature} (p2);
    \node[panel, anchor=north] at (0.65,-0.8) {(b) replace $n$ by $y$:\\ rebuilt $p'$ collides with $n$};
  \end{scope}
\end{tikzpicture}
\Description{Two small networks.  (a) Node p has fanins n and x; node n
has fanins x and y, so p and n share the fanin x.  (b) After n is
replaced by y, the rebuilt parent has exactly n's old signature, marked
as the collision.}
 }
  \caption{The divergence trigger.  (a) The self-referential shape: $p$
  and its child $n$ share the fanin $x$.  (b) After replacing $n$ by
  $y$, the rebuilt parent carries $n$'s exact signature; one engine
  rejects this substitution and the other does not.}
  \label{fig:xmg}
\end{figure}
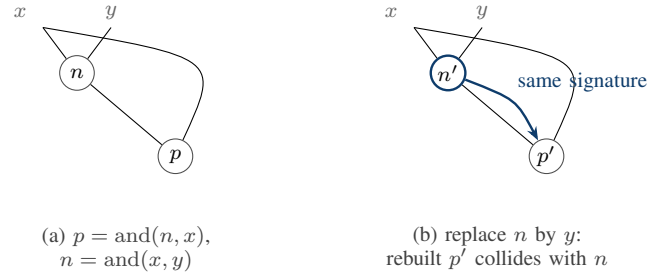

\begin{table}[t]
\centering
\scriptsize
\setlength{\tabcolsep}{3pt}
\caption{The predicate set (shortlist\ifarXiv; full set in
Appendix~\ref{app:proofspred}\fi).}
\label{tab:predicates}
\begin{tabular}{@{}p{0.25\columnwidth}p{0.34\columnwidth}p{0.31\columnwidth}@{}}
\toprule
Gate & Predicate & Contract and evidence \\
\midrule
G-P4 (R2p dedupe) & Root critical & Certified dedupe,
Thm.~\ref{thm:r2p} \\
G-P5 (algebraic) & No gap node & Certified sterile,
Thm.~\ref{thm:gap} \\
G-RW (rewrite) & Level-2 PI-fanin root & Certified sterile,
Thm.~\ref{thm:no4cut} \\
G-RF (refactoring) & Genuine two-PI AND root & Certified sterile;
\taco{}-skip only \\
G-E1E5 (cut-4 accounting) & $\mathrm{MFFC}>s^*(\mathrm{class})$ & Certified
accounting, Thm.~\ref{thm:cut4} \\
G-P6 (XMG boundary) & No trigger pair after XMG DFS & Source-derived
Tier-2; mechanism in
Thm.~\ref{thm:xmg} \\
\bottomrule
\end{tabular}
\end{table}
 \section{Cross-Carrier Witnesses and the Dynamic Boundary}
\label{sec:crosscarrier}

This section gives tight bridge costs, a witness separating
carrier-specific optima, and a scoped obstruction that delimits part
of the dynamic residue.

\subsection{Exchange rates}

Bridge actions convert between carriers, and their per-node worst-case
costs are tight.  Lowering to MIG dispatches per source node type:

\begin{center}
\scriptsize
\begin{tabular}{@{}lll@{}}
\toprule
source node & lowered form & gates/levels \\
\midrule
AND (AIG/XAG) & $\mathrm{maj}(0,a,b)$, i.e.\ an AND & $\le 1\,/\,\le 1$ \\
XOR (XAG) & 3 majority gates & $\le 3\,/\,\le 2$ \\
XOR3 (XMG, 3-input XOR) & 3 majority gates & $\le 3\,/\,\le 2$ \\
\bottomrule
\end{tabular}
\end{center}

\noindent
When no structural-hash reuse is available, these bounds equal the
MIG-native minima: the MIG NPN table's own forms for these classes have
the same cost, and for XOR2 the bound $s^* = 3$ is proved by case
analysis: no two-majority form exists.  Hash reuse can only lower the
realized cost.  An
excursion that creates $k$ XOR nodes therefore costs at most $3k$ gates
at the bridge, and the lane comparison evaluates that bound against the
excursion's measured yield.

\subsection{Complementarity}

Cross-carrier operator pairs do not dominate each other, even when they
share a name and a cut size.  AIG rewriting minimizes AIG size against
an AIG table; MIG cut rewriting minimizes MIG size against a MIG table.
Call a function \emph{MAJ-compressible} when its MIG form is strictly
smaller than its AIG form; MAJ3 is the smallest example, requiring one
majority gate versus four AIG ANDs (Fig.~\ref{fig:complementarity}).  On
such a class, a cone can remain above its MIG optimum even after AIG
rewriting has saturated.

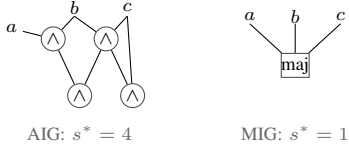
\begin{figure}[t]
  \centering
  \begin{tikzpicture}[
    gate/.style={circle, draw=black!70, inner sep=0pt, minimum size=9pt, font=\scriptsize},
    leaf/.style={font=\scriptsize},
    lbl/.style={font=\scriptsize, text=black!60}
  ]
\node[gate] (a1) at (0,0) {$\wedge$};
  \node[gate] (a2) at (0.7,0) {$\wedge$};
  \node[gate] (a3) at (0.35,-0.75) {$\wedge$};
  \node[gate] (a4) at (1.05,-0.75) {$\wedge$};
  \node[leaf] at (-0.55,0.15) {$a$};
  \node[leaf] at (0.28,0.42) {$b$};
  \node[leaf] at (0.98,0.42) {$c$};
  \draw (a1) -- (a3); \draw (a2) -- (a3); \draw (a2) -- (a4);
  \draw (-0.4,0.1) -- (a1); \draw (0.28,0.3) -- (a1); \draw (0.28,0.3) -- (a2); \draw (0.98,0.3) -- (a2);
  \draw (0.98,0.3) -- (a4);
  \node[lbl] at (0.35,-1.25) {AIG: $s^* = 4$};
\node[gate, rectangle] (m) at (3.2,-0.35) {maj};
  \node[leaf] at (2.6,0.3) {$a$};
  \node[leaf] at (3.2,0.3) {$b$};
  \node[leaf] at (3.8,0.3) {$c$};
  \draw (2.6,0.2) -- (m); \draw (3.2,0.2) -- (m); \draw (3.8,0.2) -- (m);
  \node[lbl] at (3.2,-1.25) {MIG: $s^* = 1$};
\end{tikzpicture}
\Description{Two drawings of the same three-input majority function:
four two-input AND gates on the left, one majority gate on the right.}
   \caption{The smallest MAJ-compressible gap: MAJ3 requires four AIG
  ANDs versus one MIG majority gate.}
  \Description{Two drawings of the same three-input majority function:
  four two-input AND gates on the left, one majority gate on the
  right.}
  \label{fig:complementarity}
\end{figure}

\begin{proposition}[complementarity]
\label{prop:complementarity}
MIG cut rewriting after AIG rewriting is not dominated: it offers a
positive-gain candidate on the MAJ-compressible cones of a state that
meet the accounting condition of Theorem~\ref{thm:cut4}.
Hence a predicate of the form ``AIG-side optimization makes MIG-side
NPN rewriting redundant'' fails in particular on the intersection of
the MAJ-compressible class and the fire region of
Theorem~\ref{thm:cut4}.
\end{proposition}

\noindent
The MAJ-compressible gap therefore prevents an AIG pre-pass from
substituting for MIG-side rewriting inside a MIG flow.

\subsection{Impossibility (scoped statements)}

The impossibility result is a fiber obstruction.  A mixed fiber rules
out a complete polarity-blind structural certificate; source-level
function dependence alone leaves the question open.

\begin{theorem}[polarity-erased fiber obstruction]
\label{thm:frontier}
Fix an operator $e$ at root or pass scope and let
$z_e(x)=1$ exactly when its endpoint on legal input $x$ is the identity.
If two inputs $x_0,x_1$ satisfy
$\sigma(x_0)=\sigma(x_1)$ but $z_e(x_0)\ne z_e(x_1)$, then no structural
predicate is both sound and complete for $e$'s identity region on any
domain containing both inputs.
\end{theorem}

\begin{proof}
Every structural predicate is constant on a fiber of $\sigma$, whereas
$z_e$ takes both values on the exhibited fiber.
\end{proof}

\noindent
For cut-4 rewriting, let
$m_1=\mathrm{maj}(a,b,c)$,
$m_2=\mathrm{maj}(m_1,a,b)$, and
$n=\mathrm{maj}(m_1,m_2,d)$.  With plain edges, absorption gives
$m_2=m_1$; the root lies in the one-gate majority class and cut-4
rewriting offers positive gain.  Complementing only the
$m_1\!\to m_2$ edge preserves $\sigma$, levels, fanouts, and cone sizes,
yet produces a genuine four-variable class whose stored form does not
beat the two-record cone, so the evaluator offers nothing.  The two
inputs therefore form the required mixed fiber for this deployed
operator.

A polarity-erased admission gate assigns the same decision to both
members of that fiber.  Soundness therefore requires both to undergo
dynamic evaluation.
The witness thus fixes the boundary of this predicate language for
cut-4 rewriting, while narrower structural subclasses and
function-aware gates remain available.  For the other surviving
engines, source inspection identifies the function data read by their
acceptance rules; an analogous obstruction remains open until a mixed
fiber is exhibited.
\ifarXiv Appendices~\ref{app:cut4-certificates}
and~\ref{app:proofspred} give the decoded cut-4 evidence and the full
gate ledger, respectively.\fi

Coordinate envelopes require a separate scope qualification.
\ifarXiv The full construction appears in
Appendices~\ref{app:envelopes} and~\ref{app:closedset}.\fi\ The pinned actions include
constructors that may trade nodes, depth, fanout, or path multiplicity;
the registry supplies no invariant closing a box formed from corpus
maxima.  Source-level statements therefore quantify over
$S_{\mathrm{impl}}$ or an explicitly named subset, and measurements
quantify over named corpora.  Any coordinate box remains descriptive
until an action-specific invariant closes it.
 \section{The Optimizer \taco{}}
\label{sec:optimizer}

The two layers compile into a deterministic \emph{gate layer}: before
an implemented operator is evaluated, its gate either proves that the
evaluation is unnecessary, applies a registered Tier-2 boundary
decision, or leaves the operator in the dynamic residue.  \taco{} is
the executable realization of that layer.  We evaluate three named
configurations: \taco{} in the Orchestrate comparison, the
higher-effort \taco{}-max in the HeLO comparison, and \taco{}-skip as a
gate-only diagnostic on stock \texttt{orchestrate}.  All three use
fixed deterministic schedules and selectors.

The Orchestrate-comparison form directly modifies the pinned ABC function
\texttt{Abc\_NtkOrchLocal}.
It preserves the strashed-AIG representation, topological per-root
traversal, stock RW/RS/RF engines and their source parameters, level
discipline, and the published single/iterative/\texttt{resyn}/
\texttt{resyn3} mode structure.  The gate layer is inserted at the
operator-admission points.  The retained basis also exposes R2p, a
depth-constrained instance of the same RS engine, and ranks the
remaining candidates by a fixed gain, level, and tie-breaking rule.
Consequently,
the orchestration experiment compares two selectors inside the same
published backbone.

The HeLO-comparison protocol reuses this gated AIG backbone, then adds a deterministic
native-MIG loop and carrier excursions so that the available operator
families and NDP objective match the HeLO protocol as closely as the
published interfaces allow.  Its maximum-effort setting keeps the gate
set and action menu fixed, replaces the normal \texttt{resyn} AIG
pre-pass with an iterative fixpoint of at most 20 passes, and raises the
MIG-round and time caps.  The fixed gate set shared by the two
protocol forms is split into
\begin{equation}
\begin{aligned}
\mathcal{G}_{\mathrm{exact}}
  &= \{\mathsf{G\text{-}P4},\mathsf{G\text{-}RW},
      \mathsf{G\text{-}P5},\mathsf{G\text{-}P12},
      \mathsf{G\text{-}XAGCF}\},\\
\mathcal{G}_{\mathrm{T2}}
  &= \{\mathsf{G\text{-}P1},\mathsf{G\text{-}P2},
      \mathsf{G\text{-}P3},\mathsf{G\text{-}P6}\},\\
\mathcal{G}_{\mathrm{TACO}}
  &= \mathcal{G}_{\mathrm{exact}}\mathbin{\dot\cup}
     \mathcal{G}_{\mathrm{T2}}.
\end{aligned}
\label{eq:gatetiers}
\end{equation}
Here G-P4 is the critical-root R2p dedupe; G-RW the no-4-cut RW skip;
G-P5 the MIG algebraic precondition; G-P12 the XMG-DFS identity gate;
and G-XAGCF the XOR-free XAG skip.  G-P1/G-P2/G-P3/G-P6 are
respectively the MFFC-1 cover, RF-territory boundary, source fanout
guard, and XMG boundary gate.  The first subset consists of same-scope
identities or deduplications; the second consists of source-derived
boundary decisions with registered residuals.
$\mathcal G_{\mathrm{exact}}$ denotes the five exact gates used by the
integrated \taco{} protocols.  The separate exact G-RF identity filter
is reserved for \taco{}-skip's stock path; integrated \taco{} uses
G-P1 and G-P2 for RF admission.

\taco{} uses this fixed gate set, source-default engine parameters,
and the stated deterministic stopping rules.  For the HeLO comparison,
\taco{}-max retains the
same gate set and selector while using the iterative AIG pre-pass and
larger MIG-round and time limits described above.
\taco{}-skip isolates two exact gates in the stock LocalGreedy path,
leaving its selector unchanged and enabling a whole-flow bit-identity
check.  Integrated \taco{} runs the full deterministic
gate-conditioned selector; Section~\ref{sec:experiments} reports its
integrated QoR.
In the HeLO-comparison protocol, the normal AIG-side orchestrated pass
(\texttt{resyn} mode) first reduces node count under the level
discipline; carrier excursions then evaluate the other
representations, and the NDP-evaluated loop continues in native MIG.
At maximum effort, the iterative AIG pre-pass runs for at most 20
iterations, followed by at most 60 MIG rounds under a doubled time
budget.  The extra AIG iterations can trade depth for fewer nodes, and
that trade need not improve NDP.  Maximum effort aligns the evaluation
budget more closely with HeLO while retaining the same deterministic
menu and selector.
Table~\ref{tab:forms} compares the protocol forms by dimension, and
Algorithms~\ref{alg:select} and~\ref{alg:loop} give the two selection
rules.

\begin{table}[t]
\centering
\scriptsize
\caption{Two evaluation protocols sharing the source-derived gate
registry and the Orchestrate-derived AIG backbone.}
\label{tab:forms}
\begin{tabular}{@{}l >{\raggedright\arraybackslash}p{0.34\columnwidth} >{\raggedright\arraybackslash}p{0.4\columnwidth}@{}}
\toprule
Dimension & Orchestrate comparison & HeLO comparison \\
\midrule
Backbone & Direct modification of \texttt{Abc\_NtkOrchLocal} & Same
gated AIG front end, then deterministic extension \\
Representation & Strashed AIG throughout & AIG pre-pass, carrier excursions (MIG/AIG/XAG/XMG), native MIG output \\
Objective & Node gain under the published level discipline & NDP $=$ Nodes $\times$ depth \\
Selection & Per-root: 4 candidates, gain/level-ranked commit & Per-pass: best strict-NDP-improving pass committed \\
Moves & ABC engines (rewrite, resub, R2p, refactor) & Native mockturtle passes (algebraic strategies, cut-4 rewriting, resub, balance, refactor) \\
Scheduling & Published mode structure (single, iterative, resyn, resyn3) & Fixed lane order; round loop; shared time budget \\
Higher effort & Published iterative schedule & Iterative AIG pre-pass; more MIG rounds and time; same deterministic menu \\
\bottomrule
\end{tabular}
\end{table}

In the pseudocode, $\textsc{Enabled}(P,m)$ removes only the pass
named by a registered gate whose predicate holds; a gate for one pass
never suppresses unrelated passes.

\begin{algorithm}[t]
\caption{Per-root deterministic selection (Orchestrate comparison)}
\label{alg:select}
\begin{algorithmic}[1]
\Require Strashed AIG $s$, engines $E = \{\mathrm{RW}, \mathrm{RS},
\mathrm{R2p}, \mathrm{RF}\}$
\ForAll{Roots $v$ of $s$ in topological order}
  \State $E_v \gets \{e\in E:\text{no enabled gate for $e$ fires at
  $(s,v)$}\}$
  \ForAll{$e\in E_v$}
    \State Evaluate $e$ at $v$; record $c_e=\bot$ if it offers nothing,
    otherwise $(g_e,\ell_e)$
  \EndFor
  \State $\mathcal{F} \gets \{e \in E_v : c_e\neq\bot \land
  [g_e > 0 \lor (g_e = 0 \land \ell_e < \ell(v))]\}$
  \If{$\mathcal{F} \neq \emptyset$}
    \State Commit $\arg\max_{e \in \mathcal{F}} g_e$; ties by lower
    $\ell_e$, then by $\mathrm{RW} \succ \mathrm{RS} \succ \mathrm{R2p}
    \succ \mathrm{RF}$
  \EndIf
\EndFor
\end{algorithmic}
\end{algorithm}

\begin{algorithm}[t]
\caption{Carrier lanes and NDP round loop (HeLO comparison)}
\label{alg:loop}
\begin{algorithmic}[1]
\Require AIG $s$; shared budget $B$
\State $P \gets \{\mathrm{alg}_{\mathrm{sel}},\mathrm{alg}_{\mathrm{dfs}},
\mathrm{alg}_{\mathrm{aggr}},\mathrm{cut4}\}$
\State $P \gets P\cup\{\mathrm{resub},\mathrm{bal},\mathrm{refac}\}$
\State $a_0 \gets$ common AIG after the selected AIG pre-pass
\ForAll{$L\in[\mathrm{MIG},\mathrm{AIG},\mathrm{XAG},\mathrm{XMG}]$}
  \If{$B$ is spent} \State \textbf{break} \EndIf
  \State $m \gets \textsc{LowerToMig}(\textsc{Excursion}(L,a_0))$
  \State Enable every $p\in P$
  \Repeat
    \State $P_m \gets \textsc{Enabled}(P,m)$
    \ForAll{$p \in P_m$}
      \State Evaluate $p$ on $m$; record $\mathrm{NDP}(p(m))$
    \EndFor
    \State $\mathcal{I}\gets\{p\in P_m:
    \mathrm{NDP}(p(m))<\mathrm{NDP}(m)\}$
    \If{$\mathcal{I}=\emptyset$} \State \textbf{break} \EndIf
    \State $p^* \gets \arg\min_{p\in\mathcal{I}}\mathrm{NDP}(p(m))$
    \State $m \gets p^*(m)$ (ties by fixed pass-enum order)
  \Until{$B$ is spent}
  \State Retain $m$ iff its NDP strictly beats the incumbent
\EndFor
\State \Return incumbent (lane ties retain the earlier fixed lane)
\end{algorithmic}
\end{algorithm}
 \section{Agentic Source Analysis}
\label{sec:production}

The method is a weakest-first loop over implemented operators
(Fig.~\ref{fig:method}).  Each round performs four steps:
\begin{enumerate}
\item measure offers and commits inside the deployed pipeline, because
a standalone harness can misstate an operator's role;
\item select the least-productive surviving operator;
\item derive from the pinned implementation an alias, containment
relation, or state-conditional certificate for the operator's idle
region; and
\item submit the draft relation and its source anchors to an independent
adversarial audit.
\end{enumerate}
Every accepted relation shrinks the menu examined by later rounds.
The loop stops when no further relation survives the audit.  LLM
agents generate bounded source-reading drafts.  A draft enters the
registry only after its proof obligation, source anchors, and checker,
where applicable, support it under independent audit and author
adjudication.  This organization makes a large number of bounded
readings affordable.

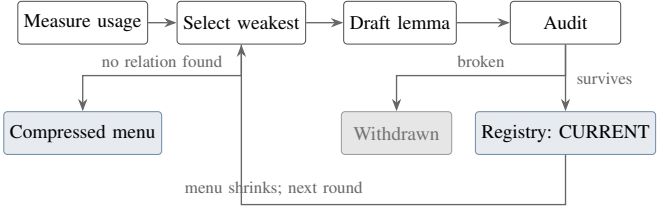
\begin{figure}[t]
  \centering
  \resizebox{\columnwidth}{!}{\begin{tikzpicture}[
    stage/.style={draw=black!70, rounded corners=1.5pt, minimum width=1.6cm,
                  minimum height=0.62cm, font=\footnotesize, inner sep=2.5pt,
                  align=center},
    outcome/.style={stage, fill=FDdark!10, draw=FDdark!60},
    dead/.style={stage, fill=black!10, draw=black!40, text=black!55},
    lbl/.style={font=\scriptsize, text=black!60},
    arrow/.style={-{Stealth[length=1.8mm]}, black!60, line width=0.55pt}
  ]
\node[stage] (measure) at (0,0) {Measure usage};
  \node[stage] (weakest) at (2.3,0) {Select weakest};
  \node[stage] (lemma) at (4.6,0) {Draft lemma};
  \node[stage] (audit) at (7.0,0) {Audit};
  \draw[arrow] (measure) -- (weakest);
  \draw[arrow] (weakest) -- (lemma);
  \draw[arrow] (lemma) -- (audit);
\node[outcome] (current) at (7.0,-1.6) {Registry: CURRENT};
  \node[dead] (broken) at (4.55,-1.6) {Withdrawn};
  \draw[arrow] (audit) -- node[lbl, right=1pt] {survives} (current);
  \draw[arrow] (audit.south) -- ++(0,-0.45) -| node[lbl, pos=0.25, above] {broken} (broken.north);
\node[outcome] (menu) at (0,-1.6) {Compressed menu};
  \draw[arrow] (weakest.south) -- ++(0,-0.45) -| node[lbl, pos=0.25, above] {no relation found} (menu.north);
\draw[arrow] (current.south) -- ++(0,-0.72) -| node[lbl, pos=0.45, above] {menu shrinks; next round} (weakest.south);
\end{tikzpicture}
\Description{A flowchart loop: measure usage, pick the weakest
operator, draft a lemma, audit it; surviving lemmas enter the CURRENT
registry and the loop repeats on the shrunken menu; the loop exits to
the compressed menu when no further relation is found.}
 }
  \caption{The agentic source analysis loop.  Usage is measured in the
  deployed pipeline, the weakest operator is examined next, and every
  draft lemma must survive an adversarial audit over the cited source
  lines before it enters the citable registry.}
  \Description{See figure.}
  \label{fig:method}
\end{figure}

The audit protocol makes each claim traceable.
Table~\ref{tab:agentprotocol} lists the record required to replay one
analysis task.  A registry entry
contains the statement, semantic scope, pinned source anchor, proof
obligation, status, and an executable checker where the obligation is
finite.  Acceptance depends on that obligation and its attached
evidence.

\begin{table}[t]
\centering
\scriptsize
\caption{The reproducible record for one agentic source-analysis task.
Replay follows the stated obligation and its attached evidence.}
\label{tab:agentprotocol}
\begin{tabular}{@{}p{0.22\columnwidth}p{0.70\columnwidth}@{}}
\toprule
Record & Required content \\
\midrule
Task packet & Fixed claim or operator, pinned source tree, semantic
scope, and instruction to seek a counterexample \\
Draft & Statement, source anchors, proof obligation, and any finite
checker \\
Independent audit & Separate read-only assignment instructed to
refute; findings preserved even when the claim is repaired \\
Acceptance & Author adjudication; only a supported final statement may
enter the registry as CURRENT \\
Release & Task prompt, run log, registry, per-claim audit records,
checker manifest, and final digests \\
\bottomrule
\end{tabular}
\end{table}

Claims that survive independent re-reading enter the registry as
CURRENT; WITHDRAWN and REFUTED entries remain visible as provenance.
Finite obligations receive exhaustive checkers.
\ifarXiv Appendices~\ref{app:failures}, \ref{app:validation},
and~\ref{app:refutations} report the failure controls,
final-validation ledger, and representative refutations.\fi
The retained audits refuted or narrowed draft claims, corrected a tie
rule, and rejected a defective XMG predicate implementation before
registration.  The preserved trail records each change.
 \section{Experiments}
\label{sec:experiments}

We use the four named sets in Table~\ref{tab:benchsets}.
\begin{table}[t]
\centering
\small
\caption{Benchmark sets and their roles.}
\label{tab:benchsets}
\setlength{\tabcolsep}{3pt}
\begin{tabular}{@{}p{0.10\columnwidth}
                    >{\raggedright\arraybackslash}p{0.84\columnwidth}@{}}
\toprule
$D_{55}$ & Supplied development manifest: 10 EPFL arithmetic, 10 EPFL
random/control, 11 ISCAS'85, 12 official IWLS-2026 seed AIGs, and 12
randomly generated AIGs (\texttt{random-01}--\texttt{04} and
\texttt{h016}--\texttt{h023}). \\
$O_{16}$ & The orchestration paper's 16 cases: five EPFL, four
VTR/FlowTune, two ISCAS'89, and five ITC'99; five overlap $D_{55}$. \\
$U_{66}$ & $D_{55}\cup O_{16}$, the 66 unique inputs used for the
gate-only and corpus-wide checks. \\
$H_{6}^{\dagger}$ & \texttt{aes\_core}, \texttt{chip\_bridge},
\texttt{fpu}, \texttt{i2c}, \texttt{mem\_ctrl}, and \texttt{s38417};
the per-case HeLO comparison. \\
\bottomrule
\end{tabular}
\vspace{2pt}

\begin{minipage}{0.96\columnwidth}
\footnotesize $^\dagger$Input status: exact: \texttt{aes\_core},
\texttt{i2c}, \texttt{mem\_ctrl}; LSOracle-OPDB reconstruction (Yosys
plus ABC latch cut): \texttt{chip\_bridge}, \texttt{fpu}; disclosed
near variant: \texttt{s38417}.  Of HeLO's other seven rows, five have
no reconstructible reported input (\texttt{pico-rv}, \texttt{des\_perf},
\texttt{ethernet}, \texttt{vga\_lcd}, \texttt{DMA}) and two have no
packaged input (\texttt{dyn\_node}, \texttt{fpga\_bridge}); none enters
a geomean.
\end{minipage}
\end{table}
i2c and mem\_ctrl belong to both $U_{66}$ and $H_6$.  The four sets
cover every reported experiment.  An older 11-case gate-development
record appears only where it is labeled historical.
The source-level theorems quantify over their implementation domains.
For performance, $D_{55}$ is the development manifest.  The eleven cases
in $O_{16}\setminus D_{55}$ form a disjoint evaluation subset; their
corrected final artifacts were added after the method
was frozen and are reported separately below.  The $H_6$ study
is a fixed-case protocol comparison.

The two comparisons follow the optimizer's two protocol forms.
Orchestrate supplies a same-source controlled comparison because
\taco{} is derived directly from its implementation.  The HeLO study
uses every reported case for which a usable input can be recovered and
separates exact inputs from reconstructions and the near variant.  The
reported performance metrics are logic-level nodes, depth, and NDP in
the pinned implementations.

The optimizer has three named configurations:
\begin{itemize}
\item \taco{}: the gate set $\mathcal G_{\mathrm{TACO}}$ of
Eq.~\eqref{eq:gatetiers}, source-default engine parameters, and the
stated deterministic stopping rules;
\item \taco{}-max: the same gates and selector, an iterative AIG
pre-pass of at most 20 passes, and larger MIG-round and time limits;
\item \taco{}-skip: the gate-only diagnostic, comprising the published
orchestration algorithm with
two trajectory-certified G-RW/G-RF gates added.
\end{itemize}
The orchestration comparison uses \taco{} under the published mode
schedules.  The HeLO comparison uses \taco{}-max to match the broader
operator-evaluation budget of that protocol.  Every reported result
records its effort (evaluations, passes, seconds on one machine).
Every optimizer output passes one final untimed equivalence sweep.
Reference outputs come from the pinned binaries under their published
commands.

\subsection{Validation of the compression hypothesis}

\subsubsection{The predicate layer at work}
The gated variant of the published orchestration algorithm
(\taco{}-skip) produces \emph{bit-identical outputs on all 66 circuits}
and, in the bit-identity harness (isolated reruns of the pinned
binary), reduces total runtime from 38.2\,s to 34.0\,s ($-11\%$).  The
two gates' savings overlap almost completely: each alone recovers 4.1\,s
of the 4.2\,s (Table~\ref{tab:gates}).  Their shared near-PI firing
region explains the overlap, while each gate suppresses a different
engine evaluation.

\begin{table}[t]
\centering
\scriptsize
\caption{Predicate-layer savings, same algorithm, 66 circuits.}
\label{tab:gates}
\begin{tabular}{@{}llll@{}}
\toprule
Configuration & Outputs & Total seconds & Saved \\
\midrule
\texttt{orchestrate}~\cite{li2024orchestration} & 66/66 & 38.2 & 0\% \\
\taco{}-skip (both gates) & 66/66 identical & 34.0 & 11\% \\
\quad No-4-cut skip alone & 66/66 identical & 34.1 & 11\% \\
\quad RF 2-PI-cone skip alone & 66/66 identical & 34.1 & 11\% \\
\bottomrule
\end{tabular}
\end{table}

\subsubsection{Empirical checks of predicted regions}
The exact predictions and the measured Tier-2 boundary checks are
reported separately in Table~\ref{tab:predictions}.  On critical roots,
depth-preserving and free resubstitution coincide, and the
depth-preserving variant contributes no distinct commit.  On slack
roots where the free candidate raises the root level, the restricted
evaluator is live:
five commits on voter and four on ctrl.
The level-gap guard suppresses algebraic passes only on gap-free states.
The cut4 pass fires only where the replacement's unreferenced count
falls below the cone size, exactly the region the accounting theorem
predicts.
For XMG, Theorem~\ref{thm:xmg} predicts the collision mechanism, while
the cheaper boundary-only G-P6 shortcut is evaluated empirically in the
integrated flow: its logged trigger and outcome fields show
zero consistency violations on $H_6$.

\begin{table}[t]
\centering
\scriptsize
\caption{Exact predictions and measured Tier-2 boundary checks.}
\label{tab:predictions}
\begin{tabular}{@{}p{0.25\columnwidth}p{0.32\columnwidth}p{0.32\columnwidth}@{}}
\toprule
Theorem & Prediction & Measured \\
\midrule
R2p dedupe (Thm.~\ref{thm:r2p}) & $\equiv R$ on critical roots & 0 distinct commits \\
R2p, restricted offer & Fires only on slack roots with a level-raising free candidate & voter 5, ctrl 4 commits \\
level-gap (Thm.~\ref{thm:gap}) & Guard lossless & skips only gap-free states \\
XMG boundary (G-P6) & Mechanism requires a trace collision & 0/6 consistency violations \\
cut4 (Thm.~\ref{thm:cut4}) & Fires in MFFC$\times s^*$ region & region confirmed \\
\bottomrule
\end{tabular}
\end{table}

\subsubsection{The residue is dynamic}
The surviving restructuring operators offer broadly and win sparsely.
Across the HeLO set, balancing and Akers refactoring offer improvements
but commit rarely (on the final campaign record, balancing offers 26
times and is never selected), while the algebraic strategies and
resubstitution carry the realized gains.  This sparsity motivates
dynamic evaluation of the residue.  The mixed-fiber witness of
Theorem~\ref{thm:frontier} supplies the separate obstruction for cut-4
rewriting.

\subsection{The orchestration comparison}

The baseline is the published Orchestrate
algorithm~\cite{li2024orchestration}, invoked by the ABC command
\texttt{orchestrate}.
\taco{} was developed directly from that
implementation: \texttt{Abc\_NtkCgoPass} retains the
\texttt{Abc\_NtkOrchLocal} per-root skeleton and stock operator
implementations, then inserts the source-derived gate layer and the
fixed residual rule.  Across all rows, the strashed AIG, topological
root traversal, RW/RS/RF engine code, published mode schedules, inputs,
host, and harness are fixed.  R2p is the same RS engine with a
required-level override.

Table~\ref{tab:orchcontrol} gives the intervention ladder.
Table~\ref{tab:gates} isolates two exact admission gates with the
published selector unchanged.  The integrated \taco{} row measures the
gate-conditioned selector, including R2p admission and its
fixed gain/level/tie rule.
Table~\ref{tab:orchcases} reports the sixteen circuits from the
original paper in single-pass mode.  All seconds in this
subsection come from the portfolio harness, which times every flow the
same way; machine conditions differ between the two harnesses, so
seconds are comparable only within a table.

\begin{table}[t]
\centering
\scriptsize
\caption{Intervention ladder and attribution in the controlled
Orchestration comparison.  Ratios are nodes/levels/NDP.}
\label{tab:orchcontrol}
\begin{tabular}{@{}p{0.20\columnwidth}p{0.39\columnwidth}p{0.30\columnwidth}@{}}
\toprule
System & Increment & Measured outcome \\
\midrule
\texttt{orchestrate} & Published
\texttt{Abc\_NtkOrchLocal} skeleton and selector & Reference \\
\taco{}-skip & G-RW/G-RF at stock admission points & 66/66 identical;
$-11\%$ same-algorithm time \\
\taco{} & $\mathcal G_{\mathrm{TACO}}$, R2p admission, fixed
residual rule; source defaults and stated termination & $0.990/0.968/0.958$;
$2.6\times$ \\
\bottomrule
\end{tabular}
\end{table}

The single-pass mode is the primary system result because it runs one
gate-conditioned per-root pass before the outer \texttt{resyn}
interleavings.  \taco{} uses fewer nodes on fourteen of the sixteen
(geomean ratios $0.990$
nodes, $0.968$ levels, $0.958$ NDP) at $2.6\times$ the speed, and
\taco{} uses 16.5 seconds in total.  The one regression is
\texttt{sqrt}, at $1.2\%$ more nodes.  Disabling each gate
individually leaves that number unchanged, localizing the residual
difference to the commit rule.  We conjecture that the difference comes
from zero-gain moves accepted by the original selection
strategy and declined by ours.  The other published modes expose their
schedule tradeoffs directly.  On $O_{16}$, iterative \taco{} gives
NDP $0.989$ while using $1.020$ nodes; \texttt{resyn} and
\texttt{resyn3} give NDP $1.054$ and $1.072$.  On $U_{66}$, the
corresponding \texttt{resyn}/\texttt{resyn3} NDPs are $0.997$ and
$1.002$.
\ifarXiv The four-mode summaries appear in
Appendix~\ref{app:experiments}.\fi\
On the disjoint $O_{16}\setminus D_{55}$ subset, \taco{}
uses fewer nodes on 10/11 cases, with geomean ratios $0.990$ nodes,
$0.957$ levels, and $0.947$ NDP.  The eleven disjoint cases therefore
reproduce the 16-case trend independently of the five development
overlaps.
The Orchestrate paper~\cite{li2024orchestration} also reports
LocalGreedy-tuned \texttt{resyn} variants on five circuits.  On those
five cases, its best reported flows use 0.5--1.8\% fewer nodes than
\taco{}.
\ifarXiv Appendix~\ref{app:experiments} gives that scoped
comparison.\fi

\begin{table*}[t]
\centering
\setlength{\tabcolsep}{4pt}\footnotesize
\caption{Orchestration comparison on the sixteen circuits reported with
\texttt{orchestrate}~\cite{li2024orchestration}, single-pass mode:
per-case nodes, depth, NDP, and seconds; node and runtime ratios appear
in parentheses.  Best QoR value per case and metric is bold.  The
closing rows give geomean QoR ratios and total seconds.}
\label{tab:orchcases}
\begin{tabular}{@{}l rrrr rrrr@{}}
\toprule
 & \multicolumn{4}{c}{orchestrate~\cite{li2024orchestration}}
 & \multicolumn{4}{c}{\taco{}} \\
\cmidrule(r){2-5}\cmidrule(r){6-9}
Case & n & d & ndp & s & n (ratio) & d & ndp & s (ratio) \\
\midrule
div & 41198 & \textbf{4372} & $1.80{\times}10^8$ & 2.3
 & 41112 (\textbf{0.998}) & \textbf{4372} & {\boldmath$1.80{\times}10^8$} & 1.3 (0.547) \\
hyp & 207505 & 24803 & $5.15{\times}10^9$ & 27.4
 & 207233 (\textbf{0.999}) & \textbf{24793} & {\boldmath$5.14{\times}10^9$} & 4.3 (0.157) \\
mem\_ctrl & 45809 & \textbf{114} & $5.22{\times}10^6$ & 2.0
 & 45687 (\textbf{0.997}) & \textbf{114} & {\boldmath$5.21{\times}10^6$} & 1.1 (0.566) \\
sqrt & \textbf{19431} & \textbf{5058} & {\boldmath$9.83{\times}10^7$} & 1.4
 & 19655 (1.012) & \textbf{5058} & $9.94{\times}10^7$ & 0.7 (0.490) \\
voter & 9516 & 63 & $6.00{\times}10^5$ & 0.7
 & 8989 (\textbf{0.945}) & \textbf{61} & {\boldmath$5.48{\times}10^5$} & 0.4 (0.477) \\
bfly & 26106 & 97 & $2.53{\times}10^6$ & 0.9
 & 25734 (\textbf{0.986}) & \textbf{95} & {\boldmath$2.44{\times}10^6$} & 1.2 (1.288) \\
dscg & 25300 & \textbf{91} & {\boldmath$2.30{\times}10^6$} & 0.9
 & 25077 (\textbf{0.991}) & 92 & $2.31{\times}10^6$ & 1.4 (1.556) \\
fir & 24894 & 94 & $2.34{\times}10^6$ & 0.9
 & 24608 (\textbf{0.989}) & \textbf{89} & {\boldmath$2.19{\times}10^6$} & 1.2 (1.384) \\
syn2 & 26950 & 90 & $2.43{\times}10^6$ & 0.9
 & 26709 (\textbf{0.991}) & \textbf{80} & {\boldmath$2.14{\times}10^6$} & 1.4 (1.573) \\
s35932 & \textbf{8561} & 13 & $1.11{\times}10^5$ & 0.3
 & \textbf{8561} (1.000) & \textbf{11} & {\boldmath$9.42{\times}10^4$} & 0.1 (0.462) \\
s38584 & 10365 & 33 & $3.42{\times}10^5$ & 0.4
 & 10222 (\textbf{0.986}) & \textbf{29} & {\boldmath$2.96{\times}10^5$} & 0.2 (0.519) \\
b17\_1 & 24209 & \textbf{93} & $2.25{\times}10^6$ & 0.8
 & 24106 (\textbf{0.996}) & \textbf{93} & {\boldmath$2.24{\times}10^6$} & 0.6 (0.729) \\
b18\_1 & 65250 & 132 & $8.61{\times}10^6$ & 2.3
 & 64606 (\textbf{0.990}) & \textbf{131} & {\boldmath$8.46{\times}10^6$} & 1.8 (0.805) \\
b20 & 10093 & \textbf{67} & $6.76{\times}10^5$ & 0.4
 & 9970 (\textbf{0.988}) & \textbf{67} & {\boldmath$6.68{\times}10^5$} & 0.2 (0.522) \\
b21 & 10203 & \textbf{67} & $6.84{\times}10^5$ & 0.4
 & 10061 (\textbf{0.986}) & \textbf{67} & {\boldmath$6.74{\times}10^5$} & 0.3 (0.707) \\
b22 & 15103 & \textbf{69} & $1.04{\times}10^6$ & 0.6
 & 14862 (\textbf{0.984}) & \textbf{69} & {\boldmath$1.03{\times}10^6$} & 0.3 (0.586) \\
\midrule
Geomean ratio & \multicolumn{4}{c}{1.000}
 & (\textbf{0.9897}) & 0.9681 & 0.9581 & (0.6663) \\
Total seconds & \multicolumn{4}{r}{42.7}
 & \multicolumn{4}{r}{\textbf{16.5}} \\
\bottomrule
\end{tabular}
\end{table*}

\subsection{The HeLO comparison}

HeLO provides reported reference values without a released
implementation.  \taco{}'s deterministic HeLO-comparison protocol uses the same
Orchestrate-derived backbone, matches the available carrier/operator
families and NDP objective, and evaluates every whole-circuit operator
outcome.  \taco{}-max brings the deterministic evaluation effort
closer to HeLO by replacing the normal \texttt{resyn} pre-pass with an
iterative AIG fixpoint and raising the MIG-round and time limits.  It
retains the gate set, selector, and action menu.

Input fidelity is shown explicitly.  The primary direct comparison is
the three exact-input rows.  On those rows, \taco{}-max has an NDP
geomean ratio of $0.903$; the two reconstructed rows and disclosed
near-variant row give $0.900$ when included as context.  Table~\ref{tab:helo}
separates the three input statuses, and
\ifarXiv Table~\ref{tab:helomax} gives the per-case maximum-effort
results in a compact form.\fi\ The iterative pre-pass can trade levels
for nodes on some cases.

\begin{table*}[t]
\centering
\setlength{\tabcolsep}{5pt}\footnotesize
\caption{HeLO results grouped by input fidelity: exact inputs are the
primary direct comparison; reconstructions and the near variant are
additional context.  Entries give nodes, depth, NDP (in $10^4$), and
seconds, with ratios vs HeLO.}
\label{tab:helo}
\begin{tabular}{@{}l rrrr rrrr@{}}
\toprule
 & \multicolumn{4}{c}{HeLO~\cite{pu2025helo}}
 & \multicolumn{4}{c}{\taco{}-max} \\
\cmidrule(r){2-5}\cmidrule(r){6-9}
Case & n & d & ndp & s$^\dagger$
 & n (ratio) & d (ratio) & ndp (ratio) & s (ratio) \\
\midrule
\multicolumn{9}{@{}l}{\emph{Exact input artifacts: primary comparison}} \\
aes\_core & 21867 & 18 & 39.4 & 63
 & 19237 (0.880) & 18 (1.000) & 34.6 (0.880) & 1941 (30.8) \\
i2c & 1385 & 8 & 1.11 & 21
 & 1329 (0.960) & 8 (1.000) & 1.06 (0.960) & 2 (0.10) \\
mem\_ctrl & 56592 & 61 & 345 & 337
 & 64119 (1.133) & 47 (0.770) & 301 (0.873) & 986 (2.93) \\
\midrule
\multicolumn{9}{@{}l}{\emph{LSOracle-OPDB reconstructions: context}} \\
chip\_bridge & 58317 & 19 & 111 & 262
 & 62064 (1.064) & 16 (0.842) & 99.3 (0.896) & 301 (1.15) \\
fpu & 68099 & 20 & 136 & 480
 & 63956 (0.939) & 19 (0.950) & 122 (0.892) & 481 (1.00) \\
\midrule
\multicolumn{9}{@{}l}{\emph{Disclosed near variant: context}} \\
s38417 & 9522 & 16 & 15.2 & 30
 & 8060 (0.846) & 17 (1.062) & 13.7 (0.899) & 21 (0.70) \\
\midrule
Exact-input geomean & \multicolumn{4}{c}{1.000}
 & 0.9853 & 0.9168 & \textbf{0.9032} & 2.0476 \\
All-available context & \multicolumn{4}{c}{1.000}
 & 0.9653 & 0.9319 & \textbf{0.8996} & 1.3826 \\
\bottomrule
\end{tabular}
\vspace{2pt}

\begin{minipage}{0.98\textwidth}
\footnotesize $^\dagger$Runtime note: HeLO seconds are reproduced from
its paper and were measured on a different machine from ours.  The
runtime ratios provide cross-machine order-of-magnitude context.
Same-machine speedups are reported in the orchestration comparison.
\end{minipage}
\end{table*}

\ifarXiv
\section{Discussion}
\label{sec:discussion}

\subsection{Evidence scopes and effort settings}

Exact certification in this paper applies to a named transformation, semantic
scope, pinned implementation, and proof obligation.  A recipe identity
is certified over its stated implementation domain; a root gate is
certified only for that root evaluator; a corpus equality is a
measurement.  Moving between these scopes requires a separate theorem
or measurement.

\taco{} denotes the fixed registry-backed gate set, source-default
engine parameters, and the stated deterministic stopping rules.  Its
Tier-2 rows retain their
registered residuals.  For the HeLO comparison, \taco{}-max replaces
the normal \texttt{resyn} AIG pre-pass with an iterative fixpoint and
raises the MIG-round and time limits.  Whole-flow bit identity is claimed only for
\taco{}-skip, where the two added gates are individually exact.

Implementation domains require the same care.  The registry contains no
closure theorem for a corpus-derived box in node count, depth, fanout,
or path count.  Several pinned actions deliberately trade one of these
coordinates for another.  Corpus maxima therefore remain descriptive;
a closure proof requires action-specific invariants.

\subsection{Why a rare survivor remains}

The cleanest boundary case is selective algebraic rewriting.  On the
historical 11-case record, the no-area selective variant wins once,
on s38417, because the area-allowed aggressive variant usually offers
more.  There it changes the NDP ratio to HeLO's reported value from
1.036 without the selective pass to 0.900 with it.

The registered relations retain that variant behind a cheap gate.  The
31-action result is the exact one-recipe Pareto cover of
Definition~\ref{def:paretocover}; further relations among its survivors
remain open unless explicitly ruled out.

\subsection{Why polarity-erased structure can be insufficient}

Section~\ref{sec:crosscarrier} gives the three-gate cut-4 witness.  The
obstruction is operator-specific: the pair preserves the
entire polarity-erased graph while pinned NPN storage, absorption, and
sharing rules produce opposite evaluator outcomes.  Any
polarity-erased admission gate assigns the pair one decision, so
soundness sends both members to dynamic evaluation.

The witness fixes the structural boundary for cut-4 rewriting.
The other engines also read function data, and their sparse measured
wins motivate dynamic evaluation.  Their impossibility status remains
open until a mixed fiber is exhibited.  Cheap function-aware predicates
and additional structural subclasses remain open as well.

\subsection{Three reusable proof objects}

The implementation-level semantic results form three connected proof
objects: fixed-scope quotient relations, sound state-conditional
admission predicates, and witnessed mixed-fiber obstructions.  These forms apply
to any optimizer satisfying their premises; the pinned ABC and
mockturtle trees make those premises checkable and turn the resulting
relations into executable gates.

Some premises, such as a discarded return or a dead guard, are
individually simple.  Their value comes from operator-wide accounting:
the analysis fixes semantic scope, composes aliases and dominance into
an exact action cover, separates non-collapsible look-alikes, and
compiles the result into a controlled optimizer intervention.

\subsection{Trust and future verification}

Some behavior-determining facts, such as a discarded return value, a
live guard, or an accounting counter, appear only in implementation code.
Agentic source analysis makes that code searchable at lemma scale, and
the registry makes every accepted claim traceable to its obligation and
source anchor.  Agents reduce the labor required for a large collection
of small, source-specific obligations.

The remaining verification problem is how to reduce the human
validation cost.  Three directions appear most
useful:
\begin{itemize}
\item verifier agents that receive formal claims and fixed source
anchors;
\item machine-checked proofs for the implementation-independent
lemmas; and
\item executable certificates that turn finite obligations into
reproducible checks.
\end{itemize}
 \fi

\section{Conclusion}
\label{sec:conclusion}

Pinned logic-optimization implementations admit three reusable proof
objects for theory-derived compression.  A fixed-scope quotient maps
the 40 deployed recipe actions to a 31-action exact one-recipe Pareto
cover; sound root- and pass-scoped predicates license certified gates;
and an operator witness turns the mixed-fiber obstruction into a
concrete admission boundary.

Built on those results, \taco{} is a deterministic optimizer with an
exact core and registered Tier-2 boundary gates.  Its main executable
contribution is the gate layer.  \taco{} directly modifies the
published ABC Orchestrate backbone, while \taco{}-max extends that
deterministic backbone to the HeLO comparison's carrier/operator
protocol.

\taco{}-skip accelerates the published orchestration algorithm while
preserving bit-identical outputs.  In the controlled same-backbone
comparison, integrated \taco{} improves the orchestration result.  The
HeLO summary uses exact-input rows as the primary comparison and
reports reconstructed and near-variant rows separately for context.

All guarantees remain tied to their stated scopes.  Further relations
among the 31 survivors remain open, and the integrated-system results
measure \taco{}'s integrated gate-conditioned selector.  \taco{}-max
retains the gate set and selector while using an iterative AIG
pre-pass and larger MIG-round and time limits.

\bibliographystyle{IEEEtran}
\bibliography{references}

\ifarXiv
\appendix
\section{Proofs of the Universal Layer}

\subsection{The discarded-cut lemma}
\label{app:proofsuniv}

\begin{lemma*}[Lemma~\ref{lem:discarded}, restated]
Let $\rho$ be any recipe stage that invokes cut rewriting on a state
$s$ without binding the returned network.  Then $T_\rho(s) = s$.
\end{lemma*}

\begin{proof}
Let $\kappa$ denote the cut-rewriting operator of the pinned
mockturtle tree, and let $s$ be the persistent network on which the
stage $\rho$ acts.  Two properties of the pinned operator delimit what
evaluating $\kappa$ can do.

\medskip
\noindent\emph{Property 1 (purity).}  The operator takes the network
by immutable reference and produces its outcome as a returned value:
its evaluation constructs and rewrites a separate network and assigns
nothing through the input.  Hence evaluating $\kappa(s)$ leaves the
persistent state equal to $s$.

\medskip
\noindent\emph{Property 2 (value guard).}  The returned value is a
function of $s$ alone: it is a rebuilt network strictly smaller than
$s$ when such a rebuild is found, and $s$ itself otherwise, a final
guard returning the input unchanged when the rebuild is not smaller.

\medskip
By Property~1, the persistent state after the evaluation of
$\kappa(s)$ is $s$.  By hypothesis, $\rho$ does not bind the returned
value, so no reachable assignment consumes it: no later stage of the
enclosing recipe reads $\kappa(s)$, and every other stage is a pure
function of the persistent state or an in-place operator on the
persistent network.  Induction on the remaining stages of the recipe
then gives
\begin{equation*}
T_\rho(s) \;=\; s .
\end{equation*}
Since the stage is an identity, a second induction on the number of
unbound cut-rewriting calls shows that deleting every such call from
a recipe word leaves the word's transition unchanged; this is the
form in which Corollary~\ref{cor:quotient} applies the lemma.
\end{proof}

\subsection{The quotient}

For reference, the inherited universe numbers its actions H00--H39:
AIG scripts H01--H05, MIG scripts H06--H08, XAG script H09, XMG script
H10, and parameter variants H11--H39.

\begin{table}[h]
\centering
\scriptsize
\caption{The recipe universe H00--H39, read off the pinned sources.
Stage names: W\O{} is a cut-rewriting call whose returned network is
discarded (an identity by Lemma~\ref{lem:discarded}); fr functional
reduction; rf refactoring; rs resubstitution; bal balancing; alg
algebraic depth rewriting; rw an assigned rewrite; xcf the XAG
constant-fanin pass.}
\label{tab:universe}
\begin{tabular}{@{}l >{\raggedright\arraybackslash}p{0.40\columnwidth} >{\raggedright\arraybackslash}p{0.42\columnwidth}@{}}
\toprule
Carrier & Scripts & Variants (what they change) \\
\midrule
AIG & H00 noop; H01 seven W\O{}, two fr, one rf; H02 bal, fr, W\O{}, rf, bal, fr, two W\O{}, bal, fr, rf, W\O{}, fr; & H11 cut-3; H12/13 zero-gain; \\
    & H03 fr, W\O{}, rf, bal, two W\O{}, bal, rf; H04 bal, W\O{}, rs, W\O{}, rf, fr, rs, rf, rs, four W\O{}; & H14 preserve-depth; H15 max\_pis=4; \\
    & H05 fr, bal, rw, rf, bal, two rw, bal, rf, rw, bal, fr & H16/17 resub budgets; H18 db; H19 sat.; \\
    & & H20 critical balance \\
MIG & H06 alg, then twice (alg, two W\O{}); & H21 cut-3; H22 zero-gain; H23 preserve-depth; \\
    & H07 bal, W\O{}, rf, bal, two W\O{}, bal, rf; H08 alg, bal, three W\O{}, bal, two W\O{} & H24 selective; H25/26 aggressive; \\
    & & H27 no-area; H28 multi-NPN; H29 max\_pis=4; \\
    & & H30 critical balance \\
XAG & H09 two W\O{}, bal, rf, xcf & H31 cut-3; H32 zero-gain; H33 preserve-depth; \\
    & & H34 db; H35 critical balance \\
XMG & H10 one alg & H36 selective; H37/38 aggressive; H39 no-area \\
\bottomrule
\end{tabular}
\end{table}

\begin{corollary*}[Corollary~\ref{cor:quotient}, restated]
Eleven of the 40 actions collapse into three alias classes
$\{H04,H11,H18,H19\}$, $\{H06,H21,H22,H23,H28\}$, $\{H09,H33\}$,
removing eight; one further action is dominated; the remaining 31 form
an exact one-recipe Pareto cover.
\end{corollary*}

\begin{proof}
Write $q_h(C) = (N_h(C), D_h(C))$ for the (nodes, depth) pair of the
endpoint $T_h(C)$ of action $h$ on input $C$, write $\le$ for the
coordinatewise order on such pairs, and write $\mathrm{PF}_{A}(C)$ for
the Pareto frontier attainable from $C$ under the action set $A$.
The argument has three parts: exact aliases, one dominance, and the
count.

\medskip
\noindent\emph{Part 1: the alias classes.}  A recipe is a finite word
over stages, and a recipe variant differs from its base word only in
parameters bound to named stages (Table~\ref{tab:universe}).  Reading
the word of each class member, every differing parameter falls into
exactly one of two cases.

\emph{Case (i): the parameter binds only discarded cut-rewriting
stages.}  This covers $H11$ (cut size) and $H18$ (NPN database)
against $H04$; $H21$ (cut size), $H22$ (zero-gain acceptance), $H23$
(depth preservation), and $H28$ (multi-NPN resynthesis) against $H06$;
and $H33$ (depth preservation) against $H09$.  In each of these words
every cut-rewriting call is a W\O{} stage, hence an identity by
Lemma~\ref{lem:discarded}, so no live stage reads the parameter
and the committed transition is unchanged.

\emph{Case (ii): the parameter is the saturation flag of functional
reduction.}  This covers $H19$ against $H04$.  With the flag set, the
equivalence-substitution pass is repeated while the network size
counter changes.  The counter is the length of the raw record storage.
Substitution only marks records dead, and dead records stay allocated.
Functional reduction inserts no records: every accepted substitution
replaces a node by an existing signal or a constant.  The counter after
the mandatory pass therefore equals the
counter before it, the guard never re-fires, and the saturation loop
executes zero extra iterations: the flagged and default transitions
coincide.

In both cases the class member executes the same live stages as its
representative on every state, so, pointwise,
\begin{gather*}
T_{H04} = T_{H11} = T_{H18} = T_{H19}, \\
T_{H06} = T_{H21} = T_{H22} = T_{H23} = T_{H28}, \\
T_{H09} = T_{H33}.
\end{gather*}
No other sibling is merged by this identity argument: each remaining
variant changes a parameter that reaches at least one live stage (the
cut size of live balancing, the gain or database choice of live
refactoring, the depth or insertion budget of live resubstitution, or
the critical-only marking of live balancing).  A live read alone does
not prove that endpoints differ.  The two exhibited separations below
establish only the additional pairs named in
Proposition~\ref{prop:sharp}.

\medskip
\noindent\emph{Part 2: one dominance.}  $H00$ is the identity recipe,
$q_{H00}(C) = (N(C), D(C))$.  $H27$ bridges to the MIG carrier, runs
one no-area algebraic depth pass, and lowers back.  The imported graph
is ordered-isomorphic to $H00$'s output.  Under the no-area flag the
pass admits only single-fanout associativity moves, each replacing the
two majority nodes of an old association by at most two new majority
nodes, and no area-increasing move: a move never increases the
reachable gate count and strictly improves the rewritten root level
when it fires; subsequent cleanup cannot worsen either coordinate.
If the lane terminates abnormally, both sides return the same
incumbent.  Hence
\begin{equation*}
q_{H27}(C) \;\le\; q_{H00}(C) \qquad\text{for every input } C,
\end{equation*}
and $H00$ is dominated.  Deleting it loses no frontier point.  If
$q_{H00}(C) \in \mathrm{PF}(C)$, then $q_{H27}(C) \le q_{H00}(C)$
forces equality, since a strict inequality would contradict frontier
membership.  The same point is therefore attained after the deletion.

\medskip
\noindent\emph{Part 3: the count and the cover.}  The three classes
contain $4 + 5 + 2 = 11$ actions and contribute $3$ representatives,
removing $8$; the dominance removes one more, leaving
$40 - 8 - 1 = 31$.  At the one-recipe scope of
Definition~\ref{def:paretocover}, aliases preserve every endpoint
pointwise and Part~2 preserves the node--depth frontier across the
dominance deletion.  Minimum cardinality and arbitrary multi-step
reachable sets lie outside this result.

The two remaining look-alike pairs admit no further collapse.  Because
the MIG engine rejects an old-node structural-hash hit that the XMG
engine accepts, the two algebraic scripts reach
$q_{H08} = (3,2)$ and $q_{H10} = (4,3)$ on a seven-AND single-output
AIG.  A nine-gate witness gives $q_{H39} = (4,4)$ and
$q_{H27} = (7,5)$ for the no-area twins; the converse containment fails
when starting from the other carrier.  These witnesses refute the proposed equivalence
and cover, so all four actions are retained
(Proposition~\ref{prop:sharp}; Appendix~\ref{app:refutations}).
Thus the 31-action frontier equality is proved over every legal recipe
input.  The source-level identities and dominance above establish the
cover theorem.  Benchmark enumeration provides a finite regression of
that equality and a liveness record for the surviving representatives.
\end{proof}

\subsection{Selector lemmas}

\begin{lemma*}[Lemma~\ref{lem:selectors}(a), restated]
$E_a(s,v) \in \{E_{a'}(s,v), s\}$ for the strict-gain variant $a$ and
the zero-gain variant $a'$.
\end{lemma*}
\begin{proof}
Fix a state $s$ and a root $v$ under evaluation.  Both variants enumerate the
same ordered candidate stream $c_1, c_2, \ldots$ at $v$ under the same
gain function $g$; they differ in a single acceptance flag, the
zero-gain variant admitting candidates with $g(c_i) \ge 0$ and the
strict variant only candidates with $g(c_i) > 0$.  Each variant
commits the first candidate attaining the maximal admissible gain: a
later candidate supersedes an earlier one only on strictly larger
gain, so the first candidate is retained on a tie.  Let
$g^{*} = \max_i g(c_i)$, and let $c_{i^{*}}$ be the first candidate
with $g(c_{i^{*}}) = g^{*}$.

\emph{Case 1: $g^{*} > 0$.}  The maximum is admissible under both
flags, and no zero-gain candidate attains it, so both variants commit
$c_{i^{*}}$ and $E_a(s,v) = E_{a'}(s,v)$.

\emph{Case 2: $g^{*} \le 0$ or the stream is empty.}  The strict
variant admits nothing and returns $s$; the zero-gain variant
commits $c_{i^{*}}$ if $g^{*} = 0$ and also returns $s$ otherwise.
Either way $E_a(s,v) = s \in \{E_{a'}(s,v), s\}$.

The argument is unchanged when the engine proposes a single
deterministic resynthesis candidate at $v$ (the refactoring family):
the stream then has length one, and the two cases above are its
positive-gain and nonpositive-gain outcomes.
\end{proof}

\begin{lemma*}[Lemma~\ref{lem:selectors}(b)]
$E_a(s,v) \in \{E_{a'}(s,v), s\}$ for the no-area variant.
\end{lemma*}
\begin{proof}
Fix a state $s$ and a root $v$ under evaluation.  Both variants call the same
move routine at $v$, which sorts the children of $v$ in ascending level order
and then evaluates guards and constructors in a fixed order: shared
applicability guards (the root must be a majority node with a level
gap, the deepest child more than one level above the second-deepest,
and a level difference between the two deepest grandchildren), then
the associativity constructor, then the distributivity constructor.
The no-area flag changes exactly two things: it adds a fanout guard,
requiring the deepest child to have fanout one, and it disables the
distributivity constructor.  Write $a$ for the no-area and $a'$ for
the area-allowed variant.

\emph{Case 1: a shared guard fails.}  Neither variant constructs a
move at $v$, and $E_a(s,v) = E_{a'}(s,v) = s$.

\emph{Case 2: the shared guards pass and an associativity candidate
exists.}  If the deepest child has fanout one, the no-area guard
passes as well and both variants perform the identical two-majority
substitution, so $E_a(s,v) = E_{a'}(s,v)$.  If the fanout guard fails,
the no-area variant constructs nothing and $E_a(s,v) = s$, while the
area-allowed variant performs the substitution.

\emph{Case 3: no associativity candidate exists.}  The no-area variant
has no further constructor and $E_a(s,v) = s$, while the area-allowed
variant may perform the distributivity construction.

In every case $E_a(s,v) \in \{E_{a'}(s,v), s\}$.  The same split holds
for the compound XMG variant: its majority and XOR-association helpers
differ from the area-allowed pass by the same fanout guard and the
same disabled distribution, its mixed helper is independent of the
flag, and the three helpers' operator cases are mutually exclusive at
the root, so the conclusion is unchanged.
\end{proof}

\begin{lemma*}[Lemma~\ref{lem:selectors}(c)]
$E_a(s,v) \in \{E_{a'}(s,v), s\}$ for critical-only balancing.
\end{lemma*}
\begin{proof}
Fix a state $s$ and a root $v$ under evaluation.  The balancing pass rebuilds
the cone of each visited root from the rebalancing candidates of its
cut, choosing the lexicographic (level, size) minimum among them; the
original cone is never a candidate.  The critical-only flag changes
exactly one thing: the set of visited roots is restricted to the
marked critical path, and every off-critical node is cloned, its cone
copied verbatim.

\emph{Case 1: $v$ is marked critical.}  Both variants call the same
rebuild at $v$: the same carrier and cut size, hence the same cut and
the same rebalancing candidates, and the same lexicographic
acceptance.  They commit the same replacement, so
$E_a(s,v) = E_{a'}(s,v)$.

\emph{Case 2: $v$ is off-critical.}  The critical-only variant copies
$v$ without rebalancing, so its local endpoint is $E_a(s,v) = s$,
while the all-node variant may rebuild the cone.  Either way
$E_a(s,v) \in \{E_{a'}(s,v), s\}$.

Hence every commit of the critical-only pass is also committed by the
all-node pass at the same root, or is the identity.
\end{proof}

\begin{lemma*}[Lemma~\ref{lem:selectors}(d)]
$R_k \preceq \{R_{k+1}\}$ for insertion budgets $k \in \{0,1\}$.
\end{lemma*}
\begin{proof}
Fix a state $s$ and a root $v$ under evaluation, and let
$m = |\mathrm{MFFC}(v)|$.  The engine evaluates candidate tiers in a
fixed order of increasing insertion count: the zero-insertion tiers
(constant collapse, then equal divisor), then the one-insertion tier,
then the two-insertion tiers; it returns the first tier that yields an
admissible candidate, and the stop-gate after a failed tier reads only
the budget and $m$: a $(k{+}1)$-insertion tier is attempted only when
the budget exceeds $k$ and $m \ge k + 2$.  No tier body depends on
the budget: the divisors, the level filters, and the candidate graphs
of a tier are functions of the state alone.  The attempt sequence of $R_k$
at $v$ is therefore a prefix of the attempt sequence of $R_{k+1}$.

\emph{Case 1: $R_k$ succeeds.}  Let tier $j \le k$ be the first tier
yielding an admissible candidate $c$.  Under $R_{k+1}$ the tiers
$0, \ldots, j-1$ fail identically, their bodies being
budget-independent, and tier $j$ yields the same $c$; hence
$E_{R_{k+1}}(s,v) = E_{R_k}(s,v)$.

\emph{Case 2: $R_k$ fails every tier up to $k$.}  Then $R_k$ returns
no candidate and $E_{R_k}(s,v) = s$, while $R_{k+1}$ either succeeds
at tier $k{+}1$ or also returns $s$.  Either way
$E_{R_k}(s,v) = s \in \{E_{R_{k+1}}(s,v), s\}$.

Hence $R_k \preceq \{R_{k+1}\}$ in the sense of
Definition~\ref{def:containment}.
\end{proof}

\begin{lemma*}[Lemma~\ref{lem:selectors}(e)]
A 4-leaf refactoring cone limit adds no endpoints beyond the default
6-leaf variant.
\end{lemma*}
\begin{proof}
Fix a state $s$ and a root $v$ under evaluation.  Both variants run the same
refactoring engine at $v$ with the same resynthesis library, the same
gain mode, and the same acceptance rule; they differ only in the
applicability guard, which skips the root when the enumerated cone has
more leaves than the limit, $4$ for $a$ and $6$ for the default
variant $a'$.  Every other guard (the cone size floor, the rejection
of a resynthesis that hashes to the old root, the strictly positive
gain acceptance) is shared.  Let $\ell$ be the cone's leaf count.

\emph{Case 1: $\ell \le 4$.}  Both guards admit the root, and the
resynthesis callback receives identical inputs: the same leaves, the
same cone truth table, the same library, and the same gain mode.  It
constructs the same candidate and accepts or rejects it identically,
so $E_a(s,v) = E_{a'}(s,v)$.

\emph{Case 2: $4 < \ell \le 6$.}  The 4-leaf guard skips the root and
$E_a(s,v) = s$, while the 6-leaf variant may act.  Hence
$E_a(s,v) = s \in \{E_{a'}(s,v), s\}$.

\emph{Case 3: $\ell > 6$.}  Both guards skip the root and
$E_a(s,v) = E_{a'}(s,v) = s$.

The limit therefore affects only the applicability guard.  The default
variant applies the same resynthesis procedure to every cone admitted
by the smaller limit.
\end{proof}

\subsection{Gain-form theorems}

\begin{lemma*}[Lemma~\ref{lem:rsgain}, restated]
Any resubstitution candidate has gain at least $1$.
\end{lemma*}
\begin{proof}
Fix a state $s$ and a root $v$ under evaluation, and let
$m = |\mathrm{MFFC}(v)|$ be the size of the root's fanout-free cone.
The cone contains $v$ itself, so $m \ge 1$.  The engine evaluates
candidate tiers in increasing insertion count and reports, for a
candidate of tier $j$, the gain $g = m - j$.  We show by induction
over the reached tiers that every returned candidate has $g \ge 1$.

\emph{Tier 0: constant collapse or equal divisor.}  No node is
inserted, and the gain equals the cone size: $g = m \ge 1$.

\emph{Tier 1: one inserted node.}  This tier is evaluated only past
the first stop-gate, which blocks it whenever $m = 1$; hence
$m \ge 2$ and $g = m - 1 \ge 1$.

\emph{Tier 2: two inserted nodes.}  This tier is evaluated only past
the second stop-gate, which blocks it whenever $m = 2$; hence
$m \ge 3$ and $g = m - 2 \ge 1$.

\emph{Tier 3.}  The deployed insertion budget is two, so this tier is
unreachable; at a budget of three its stop-gate would still require
$m \ge 4$, giving $g = m - 3 \ge 1$.

Every reachable tier therefore has gain at least one, so
$g(c) \ge 1$ for every candidate $c$.  The reported
gain is, moreover, a lower bound on the realized node reduction:
insertions are hashed against existing nodes, and any such reuse only
reduces the fresh-node count below the tier value, so the committed
change obeys the same bound.
\end{proof}

\begin{lemma*}[Lemma~\ref{lem:rwgain}, restated]
Rewrite and refactoring gains satisfy $0 \le g(c) \le
|\mathrm{MFFC}(v)|$.
\end{lemma*}
\begin{proof}
Fix a state $s$ and a root $v$ under evaluation.  For a rewrite candidate
let $s(v)$ be the number of nodes the replacement would save, the
size of the cut-relative fanout-free cone defined below; for a
refactoring candidate let $s(v) = |\mathrm{MFFC}(v)|$.  In both
engines a candidate $c$ has gain $g(c) = s(v) - a(c)$, where $a(c)$
counts the fresh nodes the replacement inserts.

\medskip
\noindent\emph{Lower bound.}  The insertion counter runs with its
budget set equal to the saved count $s(v)$ and rejects the candidate
as soon as the count would exceed that budget, so every accepted
candidate has $a(c) \le s(v)$ and
\begin{equation*}
g(c) \;=\; s(v) - a(c) \;\ge\; 0 .
\end{equation*}

\medskip
\noindent\emph{Upper bound.}  Since $a(c) \ge 0$,
$g(c) \le s(v)$; for refactoring $s(v) = |\mathrm{MFFC}(v)|$ and the
bound is immediate.  For rewriting, the saved count is computed
relative to a cut $C$ of $v$: reference counts are initialized from
the current fanouts, each occurrence of a cut leaf is charged one
extra reference, and the cut-relative cone $\mathrm{MFFC}_C(v)$ is
the set of nodes whose count reaches zero in the dereference cascade
from $v$.  We claim
$\mathrm{MFFC}_C(v) \subseteq \mathrm{MFFC}(v)$.
Write $r_0(n)$ for the plain counts and $r_1(n)$ for the bumped
counts, so $r_1(n) \ge r_0(n)$ for every node $n$, and let $D_0$ and
$D_1$ be the deleted sets of the two cascades from $v$.  Both cascades
delete $v$ unconditionally and then delete a node exactly when its
accumulated decrements exhaust its count, each deletion decrementing
the node's fanins.  Consider the nodes of $D_1$ in the bumped
cascade's deletion order, and suppose every node deleted before $n$
lies in $D_0$.  The fanouts whose deletions zero $n$ in the bumped
cascade are deleted before $n$, hence lie in $D_0$ by the hypothesis,
and each of those deletions decrements $n$ in the plain cascade as
well.  The number of decrements $n$ receives in the plain cascade is
therefore at least the number it receives in the bumped one, namely
$r_1(n) \ge r_0(n)$, so $n$ is zeroed in the plain cascade and
$n \in D_0$.  By induction $D_1 \subseteq D_0$, and
\begin{equation*}
g(c) \;\le\; s(v) \;=\; |\mathrm{MFFC}_C(v)| \;\le\;
|\mathrm{MFFC}(v)| .
\qedhere
\end{equation*}
\end{proof}

\subsection{The complete universal table}

Table~\ref{tab:uvfull} records each state-independent compression used
in the action set: the relation, the mechanism that proves it, and the
stable claim ID.  Through the released artifact manifest, each ID
links to its proof, audit record, pinned source anchors, and finite
checker where applicable.  The first block is the quotient
of Corollary~\ref{cor:quotient}: three alias classes and one dominance
take the 40 recipe actions to a 31-action exact cover.
The last block lists the pairs on
which that cover is sharp (Proposition~\ref{prop:sharp}).

\begin{table}[h]
\centering
\scriptsize
\caption{The complete universal table: state-independent containments,
aliases, and covers used in the action set, with their proving
mechanisms and stable claim IDs.}
\label{tab:uvfull}
\begin{tabular}{@{}>{\raggedright\arraybackslash}p{0.28\columnwidth} >{\raggedright\arraybackslash}p{0.42\columnwidth} >{\raggedright\arraybackslash}p{0.20\columnwidth}@{}}
\toprule
Relation & Mechanism & Stable ID \\
\midrule
\multicolumn{3}{@{}l}{\emph{Recipe universe H00--H39}} \\
\midrule
Unbound cut-rewrite stage $\equiv I$ &
Discarded return value; a final guard returns the input when not
smaller &
U-DC \\
Saturation variants $\equiv$ defaults &
Append-only size counter; the saturation loop runs zero extra
iterations &
U-SAT \\
Alias classes $\{H04$, $H11$, $H18$, $H19\}$; $\{H06$, $H21$, $H22$, $H23$, $H28\}$; $\{H09$, $H33\}$ &
Differing parameters feed only identity stages
(Lemma~\ref{lem:discarded}) or the inert saturation flag &
U-ALIASES \\
$H27$ dominates $H00$ &
No-area dfs contracts fanout-1 chains; coordinatewise nonworsening &
U-H00 \\
Exact cover of 31 actions (P31) &
Classes remove eight, dominance removes H00; frontier preserved on every
input &
U-Q31 \\
\midrule
\multicolumn{3}{@{}l}{\emph{Engine variants (one engine, one flag)}} \\
\midrule
Strict-gain variants $\preceq \{\text{zero-gain}, I\}$ &
One acceptance flag; the strict variant reaches either the zero-gain
variant's endpoint or $I$ (rw$^{-}$, rf$^{-}$, B6$^{-}$, D6$^{-}$, X4$^{-}$) &
U-STRICT \\
No-area algebraic $\preceq \{\text{area-allowed}, I\}$ &
Fanout-1 parents only: a subset of the area-allowed substitutions &
U-NOAREA \\
Critical-only balance $\preceq \{\text{all-node}, I\}$ &
Visits only the critical roots under the same lexicographic
acceptance &
U-CRITBAL \\
$R_0 \preceq R_1 \preceq R_2$ (0/1-insert $\preceq$ 2-insert) &
Budgets tried in increasing order with stop-gates, so the first
successful tiers are nested &
U-RBUDGET \\
4-Leaf refactor $D_4 \preceq \{D_6, I\}$ &
Guard-only limit; admitted cones are factored by the 6-leaf variant &
U-D4 \\
3-Leaf cut rewriting dominated by 4-leaf &
The 4-leaf variant enumerates a candidate superset and commits the
strictly best gain &
U-CUT34 \\
\midrule
\multicolumn{3}{@{}l}{\emph{Kept distinct}} \\
\midrule
$H08 \not\equiv H10$ &
Old-node hash hit excluded by the MIG engine, not by the XMG one;
(3,2) vs (4,3) exhibited &
U-MIGXMG \\
$H27$ does not dominate $H39$ &
Unguarded old-node hash reverses the no-area twins; (7,5) vs (4,4)
exhibited &
U-NATWINS \\
\bottomrule
\end{tabular}
\end{table}
 \section{Proofs of the Predicate Layer}

\subsection{R2p on critical roots}

\begin{theorem*}[Theorem~\ref{thm:r2p}, restated]
Let $R$ be free resubstitution and $R^{*}$ its depth-preserving variant
with required level clamped to $\ell(v)$.  If $v$ is critical, then
$E_{R^{*}}(s,v) = E_R(s,v)$.
\end{theorem*}

\begin{proof}
Write $\rho$ for the required-level parameter of one resubstitution
evaluation at the root $v$ of $s$: the free variant $R$ runs with
$\rho = r(v)$, the required level of $v$, while $R^{*}$ overrides the
arrival-time budget so that the evaluation runs with
$\rho = \ell(v)$.  Everything else about the two runs is the same:
the same root, the same cut and leaf set, the same window cone and
its fanout-free cone, the same simulation words, and the same
candidate ladder, which tries tiers in a fixed order (constant, equal
divisor, one insertion, two insertions) and returns the first firing
tier.  The parameter $\rho$ enters the pipeline only as a level
threshold on divisor admission, in exactly four places.  A
fanout-expanded divisor $d$ collected from the window is admitted if
and only if
\begin{equation*}
\ell(d) \;\le\; \rho ;
\end{equation*}
a divisor $d$ enters the single-divisor bank if and only if
$\ell(d) \le \rho - 1$; a pair $(d_1, d_2)$ enters the two-divisor
bank if and only if $\ell(d_i) \le \rho - 2$ for $i = 1, 2$; and the
divisor combination underlying an accepted candidate must have level
at most $\rho - 1$.

Let $v$ be critical, that is, $r(v) = \ell(v)$.  Then the two runs
share the same value of $\rho$; every threshold above admits exactly
the same divisors into exactly the same banks, the ladder enumerates
the same candidates in the same order, and first-success acceptance
returns the same candidate in both runs, or none in both.  Hence
$E_{R^{*}}(s,v) = E_R(s,v)$.

On a slack root, where $r(v) > \ell(v)$, the thresholds of $R^{*}$
are strictly tighter and the two variants can diverge; the per-tier
level bounds sharpen the picture.  At $\rho = \ell(v)$ every
candidate kind has level at most $\ell(v)$: a constant has level $0$,
an equal divisor has level at most $\rho$, a one-insertion candidate
is built on a divisor of level at most $\rho - 1$, and a
two-insertion candidate on pair members of level at most $\rho - 2$.
Consequently no candidate on a critical root can raise the root
level, so the deployed gate's offer condition (the free candidate
raised the root level) is vacuous exactly on the coincidence region
and confines the evaluation of $R^{*}$ to the slack region, where
divergence is possible.  One residual is documented within that
region: divisor banks are truncated at fixed size caps, the free
variant's looser collection binds a cap strictly earlier, and the
clamped variant can then admit divisors the free variant never
reached; covering those residual wins would require evaluating $R^{*}$
on every slack root, so the deployed gate forgoes them.
\end{proof}

\subsection{Level-gap guard}

\begin{theorem*}[Theorem~\ref{thm:gap}, restated]
If no node satisfies the level-gap condition, all three algebraic
strategies act as $I$.
\end{theorem*}

\begin{proof}
Abbreviate the gap predicate of the statement by
\begin{equation*}
G(n) \;:\Longleftrightarrow\; \ell(\mathrm{deep}(n)) \;>\;
\ell(\mathrm{second}(n)) + 1 .
\end{equation*}
The three strategies differ only in which nodes they visit:
depth-first visits the maximum-level output cones, selective the
critical-path nodes, and aggressive all nodes.  They share a single
mutation primitive: each changes the network only by calling one
reassociation routine $\mathrm{rd}$ on a visited node.  At a node
$n$, the routine orders the children in ascending order of level and may mutate
the network only if all of the following hold:
\begin{itemize}
\item[(i)] $n$ is a majority gate and its deepest child
$c = \mathrm{deep}(n)$ is a majority gate;
\item[(ii)] $\ell(c) > \ell(\mathrm{second}(n)) + 1$, that is, $G(n)$;
\item[(iii)] the two deepest children of $c$ differ in level;
\end{itemize}
and the selective strategy, which may not increase area, additionally
requires (iv) that $c$ have a single fanout and (v) that one of $n$'s
two shallow children be the same node as one of $c$'s two shallow
grandchildren.  In particular, for every strategy and every node $n$,
a mutation at $n$ implies $G(n)$.

Suppose that no node of $s$ satisfies $G$, and consider any one of
the three strategies, visiting $n_1, n_2, \ldots$ in its order.  We
prove by induction on $k$ that the network after the visit to $n_k$
is $s$.  The base $k = 0$ is trivial.  For the step, the network
before the visit to $n_k$ is $s$ by the induction hypothesis, so the
routine evaluates its guard on $s$: if (i) holds then (ii) fails by
assumption, and if (i) fails the routine already stops there; in both
cases $\mathrm{rd}$ returns without mutating anything, and the network
after the visit is again $s$.  Hence each strategy terminates with
the network $s$: it acts as $I$ on $s$.

Skipping such a pass preserves the trajectory: the surrounding loop
adopts a pass's output only when it strictly improves on the
incumbent network, and an identity pass returns the incumbent
itself, so the skipped evaluation never changes the sequence of
committed states.
\end{proof}

\subsection{No-4-cut}

\begin{theorem*}[Theorem~\ref{thm:no4cut}, restated]
If $\ell(v) = 2$ and one fanin of $v$ is a primary input, the rewrite
engine proposes no candidate at $v$.
\end{theorem*}

\begin{proof}
Let $v$ be a node with $\ell(v) = 2$ and with one fanin $u$ a primary
input, and let $w$ be the other fanin.  Since $\ell(v) = 2$, the node
$v$ is an AND gate, and levels are unit cost, so
$\ell(v) = 1 + \max\{\ell(u), \ell(w)\}$ and $\ell(w) \le 1$.
Recall the cut lists the engine maintains.  A combinational input $x$
(a primary input or a constant) owns exactly one cut, the trivial cut
$\{x\}$; an AND gate $z$ with fanins $z_0, z_1$ owns
\begin{equation*}
\mathcal{C}(z) \;=\; \bigl\{\{z\}\bigr\} \;\cup\;
\bigl\{\, c_0 \cup c_1 \;:\; c_i \in \mathcal{C}(z_i),\;
|c_0 \cup c_1| \le 4 \,\bigr\},
\end{equation*}
the union running over the fanin cut lists and truncated by the
four-leaf merge cap.

Every non-trivial cut of $v$ therefore has the form $\{u\} \cup c$
with $c \in \mathcal{C}(w)$, since $\mathcal{C}(u) = \{\{u\}\}$.
There are two cases.  If $\ell(w) = 0$, then $w$ is a combinational
input, $\mathcal{C}(w) = \{\{w\}\}$, and the only merged cut is
$\{u, w\}$, with two leaves.  If $\ell(w) = 1$, then $w$ is an AND
gate whose fanins $x_0, x_1$ have level $0$ and are combinational
inputs, so $\mathcal{C}(w) = \{\{w\}, \{x_0, x_1\}\}$ and the merged
cuts of $v$ are $\{u, w\}$ and $\{u, x_0, x_1\}$.  In all cases,
every cut of $v$ has at most $1 + 2 = 3$ leaves.  The condition is
exact: if both fanins are level-$1$ AND gates, the merge
$\{x_0, x_1\} \cup \{x_0', x_1'\}$ has four leaves whenever the four
fanins are distinct, and is then evaluated.

The evaluation loop of the rewrite engine considers only cuts with
exactly four leaves: it discards every enumerated cut with fewer than $4$
leaves, and the merge cap forbids more than $4$.  Consequently the
loop evaluates no cut of $v$, and the engine returns no candidate.
The deployed
skip evaluates the two hypotheses of the theorem in constant time on
the node's cached level field; the flow maintains levels through
every commit, so the cached field equals the structural level, and it
hands the engine a structurally hashed network whose identifiers
follow topological order.  This network is produced by
structural hashing immediately before the pass, so the cached cuts
are fresh wherever the skip runs.  The skip therefore replicates the
engine's return exactly.
\end{proof}

\subsection{Cut4 certificates}
\label{app:cut4-certificates}

\begin{theorem*}[Theorem~\ref{thm:cut4}, restated]
The pass fires at $n$ through a genuine cut only if the replacement's
unreferenced gate count is below $|\mathrm{MFFC}(n)|$, and the count is
at most the stored class size $s^*$.
\end{theorem*}

\begin{proof}
The pass makes one topological sweep over the gates of the network.
At a gate $n$ it computes $V(n) := |\mathrm{MFFC}(n)|$ exactly, by
DAG-aware reference counting, and if $V(n) = 1$ it copies $n$
unchanged without consulting any cut.  Otherwise it enumerates the
cuts of $n$; truth-table minimization discards every cut whose
function ignores enough leaves to drop the support below $3$, and
each surviving (\emph{genuine}) cut $c$ is resynthesized: the cut
function is canonicalized to its NPN class, and the class's stored database
form of $s^{*}$ gates is instantiated as the replacement cone.  Write
$u(c)$ for the number of gates of that cone that are unreferenced
outside it (the unreferenced count of Table~\ref{tab:predfull}), the
replacement root counted unconditionally.  The candidate's gain and
acceptance rule are
\begin{equation*}
g(n, c) \;=\; V(n) - u(c),
\qquad
\text{$c$ accepted} \;\Longleftrightarrow\; g(n, c) \ge 1 .
\end{equation*}
Hence the pass fires at $n$ through $c$ only if
$u(c) < V(n) = |\mathrm{MFFC}(n)|$, the first claim.  The replacement
cone realizes the stored form with at most $s^{*}$ gates, and the
unreferenced count enumerates a subset of them, so $u(c) \le s^{*}$,
the second claim; and $u(c) \ge 1$ for genuine cuts, because the root
is counted unconditionally and a class of support at least $3$ stores
a genuine gate at its root.

Two consequences follow.  (i) On a one-record cone, $V(n) = 1$, the
sweep copies $n$ verbatim and offers nothing; the accounting agrees,
since $g(n, c) \le 1 - 1 = 0$.  (ii) The decoded $4$-variable table
has $222$ NPN classes with stored-size histogram
\begin{equation*}
\{0{:}2,\; 1{:}2,\; 2{:}5,\; 3{:}18,\; 4{:}42,\; 5{:}117,\;
6{:}35,\; 7{:}1\}
\end{equation*}
and the maximum stored size is $\max s^{*} = 7$, so
$g(n, c) \ge V(n) - 7$, and any node
with $V(n) \ge 8$ and a genuine cut admits an offer.  Firing requires $V(n) \ge 2$ (because
$u(c) \ge 1$), and between the bounds $2 \le V(n) \le 7$ it requires
$u(c) \le V(n) - 1$; since only $27$ of the $222$ classes have
$s^{*} \le 3$ while $88\%$ have $s^{*} \ge 4$, the firing region is
essentially $V(n) \in [2, 7]$ crossed with the cuts of the
$s^{*} \le 3$ classes.

The pass is node-count monotone overall: acceptance demands strictly
positive gain, and a final guard returns the original network
unchanged whenever the rebuild is not smaller, so the pass can never
increase the node count.  Skipping it may forfeit rare improvements
but cannot alter functional correctness.  The converse of the firing condition
fails in exactly one way, through structural hashing: when a fragment
of the replacement cone already exists in the network with references
elsewhere, that fragment is not unreferenced, so $u(c)$ drops
strictly below $s^{*}$ and the pass can fire even when
$s^{*} \ge V(n)$; only the unconditional count of the replacement
root keeps a hit on the root itself from escaping the count.
\end{proof}

\subsection{The XMG trigger}

\begin{theorem*}[Theorem~\ref{thm:xmg}, restated]
If the XMG and MIG majority components first diverge, the immediately
preceding common state contains an adjacent node--parent pair sharing a
fanin.  If no such pair occurs along their common execution trace, the
pass endpoints coincide.
\end{theorem*}

\begin{proof}
Run both algebraic passes on the same network $s$.  In the deployed
setting the XMG side carries no XOR records, since its lane is entered
from a majority-only import, so the XMG engine's XOR-guarded helpers
never fire and its pass reduces to the majority component it shares
with the MIG engine: the same reassociation rules, applied in the same
order, through the same substitution machinery.  That machinery is one
structural hash table.  When a node $n$ is replaced by an equivalent
signal $m$, every parent $p$ of $n$ is rebuilt with $m$ substituted
for $n$, and the rebuilt signature (the operator together with the
children as a multiset of full signals) is looked up in the table so
that an existing node is reused.  The engines differ in exactly one
rule of this step: if the lookup returns the node $n$ being replaced,
whose table entry is still live during the pass, the MIG engine
rejects the hit and updates the child in place, whereas the XMG
engine accepts the returned node.  The passes can therefore diverge
only at a substitution where some rebuilt parent carries exactly the
signature of the replaced node $n$; call such an event a
\emph{collision}.

A collision forces a shared fanin.  Let $n$ have children multiset
$C$ with $|C| = k$ ($k = 2$ in the AND encoding, $k = 3$ for
majority), and suppose the rebuild of a parent $p$ collides with $n$.
The rebuild has the same operator as $n$ and children multiset
$(C_p \setminus \{n\}) \cup \{m\}$, where $C_p$ is $p$'s children
multiset; equality with $C$ as multisets forces
\begin{equation*}
m \in C
\qquad\text{and}\qquad
C_p \;=\; \bigl(C \setminus \{m\}\bigr) \cup \{n\},
\end{equation*}
so $p$ already carries every child of $n$ except possibly the
substituted one: the pair $(n, p)$ shares $k - 1 \ge 1$ fanins.  The
witness shape is the self-referential pair $n = \mathrm{and}(x, y)$,
$p = \mathrm{and}(n, x)$: replacing $n$ by its child $m = y$ rebuilds
$p$ as $\mathrm{and}(y, x)$, which is $n$'s signature
(Fig.~\ref{fig:xmg}).  In majority signals the same shape is
$p = \mathrm{maj}(x, y, n)$ with $n = \mathrm{maj}(x, y, c)$, two
shared fanins, and it is functionally redundant, since
$\mathrm{maj}(x, y, \mathrm{maj}(x, y, c)) = \mathrm{maj}(x, y, c)$
by idempotence.

Consider the first step at which the two executions differ.  Before
that step their states and selected substitutions are identical
(same start, rules, and visit order), so their distinct outcomes can
arise only from the sole difference in hash-hit handling.  The step must therefore
contain a collision, and the multiset argument above implies that the
common state immediately before the step contains the stated adjacent
pair.  Contrapositively,
if no intermediate common state contains the pair, no first divergence
exists and the pass endpoints coincide.

The shape is local, but a reassociation can manufacture it while a pass
runs.  An entry- or boundary-state absence check is therefore not a
trajectory-certified pass-level dedupe gate.  P6 rechecks the shape
after the XMG DFS stage and uses it as a source-derived Tier-2 gate.
It is retained in \taco{} with this registered residual.
\end{proof}

\subsection{The complete predicate set}
\label{app:proofspred}

Table~\ref{tab:predfull} lists the predicate set deployed by integrated
\taco{} and \taco{}-skip: each gate's checkable form, its kind under
Definition~\ref{def:gates}, its stable claim ID, and its final
evidence.  The runtime names map to stable gate IDs as follows:
P4$\to$G-P4, P9$\to$G-RW, P5$'$$\to$G-P5,
P6$\to$G-P6, P12$\to$G-P12, and E1--E5$\to$G-E1E5.
G-RF is the RF 2-PI-cone skip of Table~\ref{tab:gates}.
The corresponding theorems are listed in
Table~\ref{tab:predicates}.  Where a theorem covers only a named region, the
retained residual appears in the last column.  The accounting row
bounds what the cut-4 pass can return without skipping evaluations.
Stable IDs resolve through the released artifact manifest.

\begin{table*}[t]
\centering
\scriptsize
\setlength{\tabcolsep}{3pt}
\caption{The predicate set deployed by integrated \taco{}
and \taco{}-skip.
Kinds follow Definition~\ref{def:gates}; Tier-2 rows retain separate
evidence contracts.  Stable IDs map to the released claim
manifest.  $s^{*}$ is the stored class size of
Thm.~\ref{thm:cut4}.}
\label{tab:predfull}
\begin{tabular}{@{}p{0.12\textwidth}p{0.27\textwidth}p{0.11\textwidth}p{0.09\textwidth}p{0.31\textwidth}@{}}
\toprule
Gate & Checkable predicate & Kind & Stable ID & Final evidence \\
\midrule
\multicolumn{5}{@{}l}{\emph{AIG side}} \\
P1 RW/RF skip &
Root $|\mathrm{MFFC}| = 1$ (both fanins fanout $\ge 2$ or PI) &
Tier-2 cover & G-P1 &
RS covers the stated core; two residual classes make the $H_6$ QoR
sign case-dependent in integrated \taco{} \\
P2 RF skip &
Root $|\mathrm{MFFC}| \le 2$ &
Tier-2 & G-P2 &
RS territory; blocks gain-0 depth moves and some gain-1 wins in
integrated \taco{} \\
P3 fanout guard &
Root fanout $> 1000$ (a guard in the resubstitution source) &
Tier-2 & G-P3 &
Whole-root degradation guard; neutral on all six final cases \\
P4 depth-preserving resub skip &
$v$ critical (required level $= \ell(v)$) &
Dedupe & G-P4 &
Bit-identical; no redundant evaluations on critical roots \\
P9 no-4-cut RW skip &
$\ell(v) = 2$ with a PI fanin (G-RW: also $\ell(v) \le 1$) &
Sterile & G-RW &
Bit-identical on i2c and s38417, including gate timing; 98 skips on i2c \\
G-RF identity skip &
$v$ a genuine 2-input AND of two distinct non-constant PI nodes &
Sterile & G-RF &
Used only by \taco{}-skip; G-RF alone is bit-identical on 66/66
circuits and reduces time by 11\% \\
\midrule
\multicolumn{5}{@{}l}{\emph{MIG side}} \\
P5$'$ algebraic skip &
No MAJ node with MAJ-child gap $> 1$ and grandchild split; selective adds fanout-1, shared indices &
Sterile & G-P5 &
Exact; bit-identical on all six final cases \\
P6 XMG boundary skip &
No trigger pair (adjacent node-parent fanin sharing); re-checked at each pass boundary &
Tier-2 & G-P6 &
s38417 ungated 137173 vs 136918 gated; other five cases unchanged \\
P12 XMG-dfs identity &
XMG lane entry without MAJ-precondition node (AIG bridge gives xor3-free entry) &
Sterile & G-P12 &
Exact; bit-identical on all six final cases \\
E1-E3 cut4 certificates &
E1: gate-fanins fanout $\ge 2$; E2: $|\mathrm{MFFC}(n)| = 1$; E3: no cut $\ge 3$ after tt-min. &
Sterile & G-E1E5 &
E1 registered, not implemented (low trigger rate); E2 engine-internal \\
E4-E5 cut4 fire region &
E4: fires only if value2$(c) < |\mathrm{MFFC}(n)|$, value2 $\le s^{*}$; E5: $|\mathrm{MFFC}(n)| > s^{*}$ forces an offer for a genuine cut &
Accounting & G-E1E5 &
Live on fpu, chip\_bridge, ethernet: 5 commits out of 20 offers (11-case record) \\
\midrule
\multicolumn{5}{@{}l}{\emph{XAG side}} \\
XAG-CF skip &
No XOR node in the state (xor-count zero) &
Sterile & G-XAGCF &
Identity on XOR-free states \\
\bottomrule
\end{tabular}
\end{table*}
 \section{Scoped Obstructions and Open Items}

\subsection{The polarity-erased fiber obstruction}
\label{app:fiber}

Recall that $\sigma$ retains the full typed, rooted multigraph but
erases edge polarity and primary-input names
(Definition~\ref{def:structural}).  For a
fixed operator and scope, let $z_e(x)=1$ exactly when the endpoint on
$x$ is the identity.

\begin{theorem*}[Theorem~\ref{thm:frontier}, restated]
If $\sigma(x_0)=\sigma(x_1)$ and $z_e(x_0)\ne z_e(x_1)$, no predicate
that factors through $\sigma$ is both sound and complete for $e$'s
identity region on a domain containing $x_0$ and $x_1$.
\end{theorem*}

\begin{proof}
If $\pi=\widehat{\pi}\circ\sigma$, then
$\pi(x_0)=\pi(x_1)$.  Soundness and completeness would instead require
$\pi(x_i)=z_e(x_i)$ for both inputs, which is impossible because the
two values of $z_e$ differ.
\end{proof}

\paragraph*{Exhibited cut-4 fiber}
Let
$m_1=\mathrm{maj}(a,b,c)$,
$m_2=\mathrm{maj}(m_1,a,b)$, and
$n=\mathrm{maj}(m_1,m_2,d)$.
With plain edges, majority absorption gives $m_2\equiv m_1$, so the
root is in the one-gate majority class while its removable cone has
two records; the cut-4 evaluator offers positive gain.  Complement
only the $m_1\!\to m_2$ edge.  The polarity-erased graph, levels,
fanouts, and cone sizes are unchanged, but the root becomes a genuine
four-variable class whose stored form costs at least the two-record
cone when no external hash reuse is present; the evaluator offers
nothing.  The decoded NPN table and both verdicts are independently
checked in the artifact.  Thus Theorem~\ref{thm:frontier} applies to
cut-4 rewriting.

Table~\ref{tab:fiberevidence} separates that proved mixed fiber from the
operators for which this paper establishes only source-level function
dependence.

\begin{table}[h]
\centering
\scriptsize
\setlength{\tabcolsep}{3pt}
\caption{What is proved about the function-dependent residue.
Source dependence identifies the discriminator; only an exhibited
mixed $\sigma$-fiber invokes Theorem~\ref{thm:frontier}.}
\label{tab:fiberevidence}
\begin{tabular}{@{}p{0.22\columnwidth}p{0.38\columnwidth}p{0.31\columnwidth}@{}}
\toprule
Operator & Function discriminator & Evidence in this paper \\
\midrule
Cut-4 rewriting & NPN class of the cut truth table and unreferenced stored form &
Mixed fiber exhibited; obstruction proved \\
Akers refactoring & Resynthesis of the cone truth table &
Source dependence established; mixed-fiber witness remains open \\
SOP balancing & ISOP cover of the cut truth &
Source dependence established; mixed-fiber witness remains open \\
ABC/MIG resubstitution & Truth containment/equality or window simulation &
Source dependence established; mixed-fiber witness remains open \\
ABC rewriting & Four-cut NPN class &
Source dependence established; mixed-fiber witness remains open \\
\bottomrule
\end{tabular}
\end{table}

\begin{remark}[scope]
The theorem applies where a mixed fiber is witnessed.  For the other
operators, source-level function dependence identifies the missing
information but does not by itself exhibit such a fiber.  Further
structural subclasses and inexpensive function-aware predicates also
remain open.
\end{remark}

\subsection{Limits of corpus envelopes}
\label{app:envelopes}

A natural formal route is to define a circuit class that contains the
corpus and is preserved by all forty actions, so that operator relations can be
proved once for the class.  A typical construction begins with a simple
coordinate envelope.

Let $N$, $D$, $F$, and $P$ denote node count, depth, maximum fanout,
and a fixed input--output path count, respectively.  A \emph{simple coordinate
envelope} is a finite conjunction of bounds $q(s)\le b_q$, with $q$
drawn from a subset of these coordinates.

No registered invariant establishes that the forty actions preserve a
box obtained from corpus maxima.  The pinned implementations
contain mechanisms that make such preservation nontrivial:
depth-motivated balancing may add gates, default resubstitution admits
level-raising replacements, refactoring has no depth-acceptance guard,
and divisor reuse and restructuring alter fanout and path
multiplicity.  They motivate the action-specific invariance requirement
below.

Accordingly, corpus maxima and other finite coordinate boxes remain
descriptive until accompanied by an action-specific invariance proof.
The formal statements quantify over legal implementation states or
explicitly named subsets.

\subsection{A complete closed-set example}
\label{app:closedset}

A formally complete construction makes this limitation concrete.  It
satisfies the usual set-theoretic requirements while carrying almost
no information about optimization.

\begin{example}[normal-quotient universe]
Let $\mathcal U_{\mathrm{MIG}}$ be all finite ordered combinational
MIGs.  Define a normalizer $N$ by four deterministic steps:
\begin{enumerate}
\item reduce the one-node identities
$\mathrm{maj}(x,x,y)=x$ and $\mathrm{maj}(x,\bar x,y)=y$;
\item sort fanins and structurally hash identical majority triples;
\item delete nodes unreachable from the ordered outputs; and
\item canonically relabel nodes by ASAP depth and then by their
canonical fanin triples, preserving PI and PO order.
\end{enumerate}
No reassociation, balancing, rewriting, refactoring, resubstitution,
or SAT sweeping occurs in $N$.  The procedure is deterministic,
semantics preserving, and idempotent.  Define
\[
  \mathcal S_N
  = \{C\in\mathcal U_{\mathrm{MIG}}:N(C)=C\},
  \qquad
  \widehat T_a = N\circ T_a ,
\]
where $T_a(C)$ is any committed successor of recipe action $a$ that
returns a finite ordered MIG.  Each AIG anchor enters the set through
$N\circ L_{\mathrm{AIG}\rightarrow\mathrm{MIG}}$, with
$L(\mathrm{and}(x,y))=\mathrm{maj}(0,x,y)$.
\end{example}

\begin{proposition}[closure by construction]
For every $C\in\mathcal U_{\mathrm{MIG}}$ and every admitted successor
$T_a(C)$,
\[
  \widehat T_a(C)\in\mathcal S_N .
\]
Consequently, $\mathcal S_N$ is forward closed under every normalized
transition $\widehat T_a$.
\end{proposition}
\begin{proof}
Idempotence gives
$N(\widehat T_a(C))=N(N(T_a(C)))=N(T_a(C))=\widehat T_a(C)$.
Thus the normalized successor is a fixed point of $N$.
\end{proof}

The set is also proper in the formal sense: it contains live
irreducible majority chains of unbounded size and depth.  Scalable
nonmember families result from adding unreachable logic, structurally
duplicate gates, or
$\mathrm{maj}(x,x,y)$ identity wrappers; $N$ removes each defect.
Membership is exactly decidable by normalize-and-compare, with no
fitted constants.  All circuits in the historical 55-case set belong
to $\mathcal S_N$ after normalization.

\begin{remark}[why the construction is trivial]
The proof never uses $C\in\mathcal S_N$, any property of $a$, or any
quality measure: appending $N$ to the transition forces every successor
into the set.  Moreover, $\mathcal S_N$ excludes only representation
garbage and identifier order; every Boolean function has a
representative in it.  On the 55 anchors, $N$ removed zero logic nodes,
so its only observed effect there was canonical relabeling.  The
construction is a nonempty, proper, forward-closed circuit set.  Its
closure nevertheless says nothing about which operator to run or what
quality it can achieve.  Its certificate attaches to the
normalization wrapper; deployed optimization choices require their own
relations.  Functional correctness of an
action is a separate obligation; the proposition assumes only an
admitted MIG successor.
\end{remark}

\subsection{The open register}
\label{app:open}

The scoped results above leave several concrete directions open:
\begin{itemize}
\item Mixed-fiber witnesses for the function-dependent operators listed
as open in Table~\ref{tab:fiberevidence}.
\item Function-aware predicates between structural gates and the full
engines.
\item The minimum certified basis over per-root selection modes,
currently only loosely bounded.
\item Aggressive versus selective algebraic rewriting under
finite-factor simulation in either direction.
\item Database size-optimality for stored sizes 5--7 (verified exactly
through size 4; sizes 5--7 currently rely on the database entries).
\item Ablations of each lane's role and the remaining rewrite--refactoring and
resubstitution--refactoring containment edges.
\end{itemize}
 \section{Extended Experimental Material}

\subsection{Bit-identity on all 66 circuits}

The gated variant produces outputs bit-identical to those of the published
algorithm on all 66 circuits; total runtime over the set decreases from
38.2\,s to 34.0\,s in the bit-identity harness.  No case shows a QoR
difference under either gate, alone or combined.

\subsection{All-mode geomeans against \texttt{orchestrate}}
\label{app:experiments}

Tables~\ref{tab:orch16} and~\ref{tab:orch66} report geomean ratios of
\taco{} to \texttt{orchestrate}~\cite{li2024orchestration} in each of
its four modes, over the sixteen paper circuits and over the full
66-circuit set respectively.  The gap is largest in the single-pass
mode.  The iterative mode trades about 2\% more nodes for 3--4\%
fewer levels and stays ahead on NDP.  Its commit rule continues to
accept level-reducing moves at the fixed point, producing this
node--level exchange.  The resyn interleavings run 4--7\% deeper on
the sixteen paper circuits.  Both final harnesses use the default
level-preserving engine settings.  The two flows share this setting,
while their selection and commit rules differ; on this subset, the extra
depth buys no area.  The baseline in the body is the
published Orchestrate algorithm, invoked by the ABC command
\texttt{orchestrate}.  The paper additionally reports
LocalGreedy-tuned \texttt{resyn} variants on five circuits.  For
context, those best-reported flows use 0.5--1.8\% fewer nodes than
\taco{} on the five cases.  The same mechanism appears in
\texttt{sqrt}: those flows commit zero-gain moves freely,
and our commit rule does not.  Over the full 66-circuit set the node
geomeans are at parity with or below the baseline.  Per-case numbers for every run
are released in the artifact's output files.

\begin{table}[h]
\centering
\scriptsize
\setlength{\tabcolsep}{4pt}
\caption{Geomean ratios of \taco{} to \texttt{orchestrate}, sixteen
paper circuits.  Wins count strictly smaller node counts.}
\label{tab:orch16}
\begin{tabular}{@{}lllllll@{}}
\toprule
Mode & Nodes & Levels & NDP & Wins/16 & \taco{} s & Speedup \\
\midrule
Single & 0.9897 & 0.9681 & 0.9581 & 14 & 16.5 & $2.6\times$ \\
Iterative & 1.0203 & 0.9695 & 0.9891 & 4 & 94.6 & $1.6\times$ \\
Resyn & 1.0008 & 1.0534 & 1.0543 & 8 & 41.5 & $2.7\times$ \\
Resyn3 & 1.0023 & 1.0691 & 1.0715 & 10 & 64.8 & $2.3\times$ \\
\bottomrule
\end{tabular}
\end{table}

\begin{table}[h]
\centering
\scriptsize
\setlength{\tabcolsep}{4pt}
\caption{Geomean ratios of \taco{} to \texttt{orchestrate}, full
66-circuit set.  Wins count strictly smaller node counts.}
\label{tab:orch66}
\begin{tabular}{@{}lllllll@{}}
\toprule
Mode & Nodes & Levels & NDP & Wins/66 & \taco{} s & Speedup \\
\midrule
Single & 0.9879 & 0.9515 & 0.9400 & 36 & 37.8 & $1.9\times$ \\
Iterative & 1.0001 & 0.9564 & 0.9564 & 25 & 140.7 & $1.6\times$ \\
Resyn & 0.9942 & 1.0028 & 0.9970 & 26 & 97.5 & $1.9\times$ \\
Resyn3 & 0.9944 & 1.0072 & 1.0015 & 27 & 154.6 & $1.6\times$ \\
\bottomrule
\end{tabular}
\end{table}

\subsection{Effort details}

All reported seconds are per-case wall times on one x86-64 macOS host;
totals are sums over cases.  Over the 66 circuits, the pass counts are
66 for single (one per case), 193 for iterative (at most 10 per case),
198 for resyn, and 330 for resyn3.  Per-case times appear in the body
tables and released artifact; engine-level evaluation counts are
printed in each run's \texttt{CGO\_RESULT} log.

\subsection{\taco{}-max results against HeLO, per case}

Table~\ref{tab:helomax} reports the deterministic maximum-effort
track against HeLO's published values.  E denotes an exact input, R a
reconstruction, and V the disclosed near variant; only E enters the
primary direct geomean.

\begin{table}[h]
\centering
\scriptsize
\setlength{\tabcolsep}{1.5pt}
\caption{HeLO set, per case: \taco{}-max versus HeLO's reported values.
The NDP column is for \taco{}-max, and NDP/HeLO is its ratio to HeLO's
published NDP.  E denotes an exact input, R a reconstruction, and V a
near variant.}
\label{tab:helomax}
\begin{tabular}{@{}lllllll@{}}
\toprule
Case & In & HeLO n/d & \taco{}-max n/d & NDP & NDP/HeLO & s \\
\midrule
aes\_core & E & 21867/18 & 19237/18 & $3.46{\times}10^5$ & 0.880 & 1941 \\
i2c & E & 1385/8 & 1329/8 & $1.06{\times}10^4$ & 0.960 & 2 \\
mem\_ctrl & E & 56592/61 & 64119/47 & $3.01{\times}10^6$ & 0.873 & 986 \\
chip\_bridge & R & 58317/19 & 62064/16 & $9.93{\times}10^5$ & 0.896 & 301 \\
fpu & R & 68099/20 & 63956/19 & $1.22{\times}10^6$ & 0.892 & 481 \\
s38417 & V & 9522/16 & 8060/17 & $1.37{\times}10^5$ & 0.899 & 21 \\
\midrule
Exact geomean & E & -- & -- & -- & 0.903 & -- \\
All context & E/R/V & -- & -- & -- & 0.900 & -- \\
\bottomrule
\end{tabular}
\end{table}
 \subsection{The Optional Fast Extension}
\label{app:fast}

The main-paper configuration uses source-default engine parameters and
the stated deterministic stopping rules.  The released implementation also
provides \taco{}-fast, a deterministic empirical overlay for
runtime-sensitive use.  It retains $\mathcal G_{\mathrm{TACO}}$, the
residual operator set, the selector, the lane order, and the strict-NDP
commit rule.  The overlay changes the evaluation policy through tuned
AIG bounds, an early iterative stop, and four stateful yield gates.
\taco{}-fast at maximum effort applies the same overlay with the iterative AIG
pre-pass and the larger MIG-round and time limits of \taco{}-max.

\subsubsection{Configuration and evidence contract}

Table~\ref{tab:fast-constants} gives the complete configuration delta.
The normal-effort HeLO-comparison recipe is
\texttt{cgo -M 2; cgo\_mig}; the maximum-effort recipe is
\texttt{cgo -M 1 -P 20; cgo\_mig -R 60 -T 1800}.  In the released
command line, \texttt{-C} selects the main-paper \taco{}
configuration.

\begin{table}[H]
\centering
\scriptsize
\setlength{\tabcolsep}{3.5pt}
\caption{Complete delta from \taco{} to the fast overlay.}
\label{tab:fast-constants}
\begin{tabular}{@{}lll@{}}
\toprule
Setting & \taco{} & \taco{}-fast \\
\midrule
RF node/cone bound & 10/16 & 9/14 \\
RW cut keep count & 250 & 100 \\
Iterative stop & no-improvement fixpoint & fixpoint or gain $<0.5\%$ \\
AIG yield window & disabled & 4096 evaluations \\
MIG yield window & disabled & 3 rounds \\
\bottomrule
\end{tabular}
\vspace{2pt}

\begin{minipage}{0.96\columnwidth}
\footnotesize The two engine bounds and the stopping threshold were
selected on the development corpus.  The yield windows follow the
measured dry-block convention recorded in the artifact.
\end{minipage}
\end{table}

The four online gates are listed in Table~\ref{tab:fast-gates}.  Each is
a deterministic state machine scoped to one pass or one carrier lane.
Their history-based triggers carry measured contracts, including the
residual opportunities stated below.

\begin{table}[H]
\centering
\scriptsize
\setlength{\tabcolsep}{3pt}
\caption{Empirical gates used only by the fast overlay.}
\label{tab:fast-gates}
\begin{tabular}{@{}p{0.10\columnwidth}
                    p{0.38\columnwidth}
                    p{0.42\columnwidth}@{}}
\toprule
Gate & Trigger and action & Registered residual \\
\midrule
F-P7 & After a 4096-evaluation block with zero committed gain,
disable the dry RW, RS, or RF engine for the rest of the AIG pass. &
A later block may contain a commit. \\
F-P8 & After 4096 MFFC-1 evaluations with no RS/R2p candidate,
disable RS on MFFC-1 roots for the rest of the AIG pass. &
At most a gain-1 Divs0/Quit candidate per later root is at stake. \\
F-P10 & After three offer-free cut4 rounds, disable cut4 for the rest
of that carrier lane. &
A later round may enter the cut4 fire region. \\
F-P11 & After three evaluated rounds without a win, disable the dry
resubstitution, refactoring, or balancing pass for the rest of that
carrier lane. &
A later restructuring commit may be missed. \\
\bottomrule
\end{tabular}
\end{table}

\subsubsection{Orchestration results}

Table~\ref{tab:fast-orch-modes} reports the fast overlay under all four
published Orchestration schedules.  Ratios and speedups use the same
inputs, host, harness, and published \texttt{orchestrate} command as
Tables~\ref{tab:orch16} and~\ref{tab:orch66}.  In single-pass mode, the
overlay moves the $O_{16}$ depth and NDP ratios from
$0.9681/0.9581$ to $0.9549/0.9455$, while total time decreases from
16.5\,s to 11.6\,s.  On $U_{66}$, NDP moves from $0.9400$ to
$0.9344$ and time from 37.8\,s to 16.2\,s.

\begin{table}[H]
\centering
\scriptsize
\setlength{\tabcolsep}{2.4pt}
\caption{\taco{}-fast geomean ratios to \texttt{orchestrate} under
the four published schedules.  Wins count strictly smaller node
counts; seconds are totals on the common host.}
\label{tab:fast-orch-modes}
\begin{tabular}{@{}llllllll@{}}
\toprule
Set & Mode & Nodes & Levels & NDP & Wins & s & Speedup \\
\midrule
$O_{16}$ & Single & 0.9901 & 0.9549 & 0.9455 & 14/16 & 11.6 & $3.7\times$ \\
 & Iter. & 1.0230 & 0.9630 & 0.9851 & 4/16 & 30.0 & $5.0\times$ \\
 & Resyn & 1.0028 & 1.0403 & 1.0432 & 7/16 & 26.6 & $4.2\times$ \\
 & Resyn3 & 1.0056 & 1.0551 & 1.0610 & 7/16 & 37.3 & $3.9\times$ \\
\midrule
$U_{66}$ & Single & 0.9882 & 0.9455 & 0.9344 & 36/66 & 16.2 & $4.4\times$ \\
 & Iter. & 1.0004 & 0.9524 & 0.9528 & 24/66 & 36.9 & $6.2\times$ \\
 & Resyn & 0.9955 & 1.0009 & 0.9965 & 26/66 & 36.1 & $5.2\times$ \\
 & Resyn3 & 0.9926 & 1.0023 & 0.9950 & 25/66 & 50.9 & $4.7\times$ \\
\bottomrule
\end{tabular}
\end{table}

The node geomean remains within $0.04\%$ of \taco{}, while the corrected
FlowTune/ITC inputs contribute several of the depth reductions.  On the disjoint
$O_{16}\setminus D_{55}$ subset, \taco{}/\taco{}-fast respectively
give $0.9897/0.9908$ nodes, $0.9568/0.9393$ levels, and
$0.9469/0.9307$ NDP, with 10/11 strict node wins in both configurations.
The artifact retains the per-case nodes, levels, NDP, wall time, and
engine-evaluation counts for every schedule.

The output differences follow the registered empirical choices.  On
\texttt{b18\_1} in \texttt{resyn} mode, the fast overlay is better by
115 nodes; on \texttt{s38417}, the base \taco{} configuration is better by
15 nodes.  These cases establish that the online gates can move QoR in
either direction even though the aggregate node ratios remain close.

\subsubsection{HeLO-comparison results}

Table~\ref{tab:fast-helo} reports the two \taco{}-fast effort settings on the six
available HeLO cases, preserving the input-fidelity groups of
Table~\ref{tab:benchsets}.  At normal effort, its exact-input and
all-context NDP geomeans are $0.940$ and $0.944$; at maximum effort,
they are $0.944$ and $0.929$.  For comparison, the main-paper
\taco{}-max result in Table~\ref{tab:helo} is $0.903/0.900$.

\begin{table*}[!b]
\centering
\setlength{\tabcolsep}{2.6pt}\scriptsize
\caption{Fast-overlay HeLO results grouped by input fidelity.  Entries
give nodes, depth, NDP (in $10^4$), and seconds; parenthesized values
are ratios to HeLO's published values.}
\label{tab:fast-helo}
\begin{tabular}{@{}l rrrr rrrr rrrr@{}}
\toprule
 & \multicolumn{4}{c}{HeLO~\cite{pu2025helo}}
 & \multicolumn{4}{c}{\taco{}-fast}
 & \multicolumn{4}{c}{\taco{}-fast, maximum effort} \\
\cmidrule(r){2-5}\cmidrule(r){6-9}\cmidrule(r){10-13}
Case & n & d & ndp & s$^\dagger$
 & n (ratio) & d (ratio) & ndp (ratio) & s (ratio)
 & n (ratio) & d (ratio) & ndp (ratio) & s (ratio) \\
\midrule
\multicolumn{13}{@{}l}{\emph{Exact input artifacts: primary comparison}} \\
aes\_core & 21867 & 18 & 39.4 & 63
 & 20084 (0.918) & 19 (1.056) & 38.2 (0.969) & 1301 (20.7)
 & 18367 (0.840) & 21 (1.167) & 38.6 (0.980) & 1964 (31.2) \\
i2c & 1385 & 8 & 1.11 & 21
 & 1317 (0.951) & 8 (1.000) & 1.05 (0.951) & 2 (0.10)
 & 1329 (0.960) & 8 (1.000) & 1.06 (0.960) & 2 (0.09) \\
mem\_ctrl & 56592 & 61 & 345 & 337
 & 66239 (1.170) & 47 (0.770) & 311 (0.902) & 967 (2.87)
 & 53275 (0.941) & 58 (0.951) & 309 (0.895) & 899 (2.67) \\
\midrule
\multicolumn{13}{@{}l}{\emph{LSOracle-OPDB reconstructions: context}} \\
chip\_bridge & 58317 & 19 & 111 & 262
 & 67935 (1.165) & 16 (0.842) & 109 (0.981) & 96 (0.37)
 & 67147 (1.151) & 17 (0.895) & 114 (1.030) & 164 (0.63) \\
fpu & 68099 & 20 & 136 & 480
 & 65408 (0.960) & 20 (1.000) & 131 (0.960) & 300 (0.63)
 & 59023 (0.867) & 19 (0.950) & 112 (0.823) & 460 (0.96) \\
\midrule
\multicolumn{13}{@{}l}{\emph{Disclosed near variant: context}} \\
s38417 & 9522 & 16 & 15.2 & 30
 & 8069 (0.847) & 17 (1.062) & 13.7 (0.900) & 27 (0.91)
 & 8077 (0.848) & 17 (1.062) & 13.7 (0.901) & 16 (0.54) \\
\midrule
Exact-input geomean & \multicolumn{4}{c}{1.000}
 & 1.0074 & 0.9334 & 0.9403 & 1.7804
 & 0.9121 & 1.0352 & 0.9442 & 1.9933 \\
All-available context & \multicolumn{4}{c}{1.000}
 & 0.9948 & 0.9484 & 0.9435 & 1.0296
 & 0.9289 & 1.0003 & 0.9292 & 1.1523 \\
\bottomrule
\end{tabular}
\vspace{2pt}

\begin{minipage}{0.98\textwidth}
\footnotesize $^\dagger$HeLO seconds are reproduced from its paper and
were measured on a different machine from ours.  Their ratios provide
cross-machine order-of-magnitude context; the Orchestration study above
contains the same-machine speedups.
\end{minipage}
\end{table*}

The paired normal-effort diagnostic holds the AIG pre-pass fixed and
switches only the MIG-side fast overlay.  On \texttt{fpu}, disabling
that overlay produces 59{,}794 nodes at depth 20, whereas the fast run
produces 65{,}408 nodes at the same depth.  The trace attributes the
5{,}614-node difference to F-P11: it ends resubstitution three rounds
early and forgoes two downstream depth-rewrite commits.  The same trace
supplies the concrete residual witness recorded for this empirical
gate.

\subsubsection{Scope of the fast evidence}

The tuned values were selected and evaluated on the development
manifests, so the $D_{55}$ and $H_6$ results are in-sample performance
evidence.  The eleven corrected inputs in
$O_{16}\setminus D_{55}$ entered after the configuration freeze and
form the disjoint evaluation subset reported above.  The final untimed
CEC sweep passes all 552 outputs from the main-paper and fast-overlay
configurations.  The evidence therefore covers functional equivalence,
deterministic replay, measured QoR and runtime, and the late-opportunity
residuals stated in Table~\ref{tab:fast-gates}.
 \section{Production Process: Errors and Final Validation}

\subsection{Observed failure modes}
\label{app:failures}

Three recurrent failures could promote a local observation into a
false general claim: measuring the wrong call path, citing a superseded
lemma, or accepting unsupported generated text.  Table~\ref{tab:errors}
pairs each observed instance with the control that now prevents its
reuse as evidence.

\begin{table}[H]
\centering
\scriptsize
\caption{Observed failures and their current controls.}
\label{tab:errors}
\begin{tabular}{@{}p{0.18\columnwidth}p{0.30\columnwidth}p{0.39\columnwidth}@{}}
\toprule
Failure mode & Observed instance & Countermeasure \\
\midrule
Bug--artifact confusion & Cut4 appeared sterile when a harness discarded
its returned network & Re-run through the deployed call path and check
return-value ownership \\
Citation drift & XMG was excluded using a genuine but later-withdrawn
lemma & Status registry; only CURRENT entries may support a result \\
Confabulation & An LLM produced claims unsupported by the record &
Pinned source anchors, explicit obligations, and independent adversarial
re-reading \\
\bottomrule
\end{tabular}
\end{table}

\subsection{The XMG correction, condensed}
\label{app:xmg-correction}

A standalone harness reported no wins for the XMG lane on the
11-case set.  A subsequent exclusion relied on a stale ledger claim
that the XMG and MIG algebraic scripts were cost-equivalent on every
AIG input.  The registry had already withdrawn that claim: a seven-AND witness reaches
$(3,2)$ under the MIG script and $(4,3)$ under the XMG script.  The
lane was restored; the incident motivated the status registry.

Source re-reading localized the difference to an
old-node structural-hash collision.  Theorem~\ref{thm:xmg} states the
necessary condition for a first divergence: the preceding common state
contains an adjacent node--parent pair sharing a fanin.  Since
reassociation can create the pair, its absence at pass entry does not
guarantee absence throughout the pass.

P6 checks the boundary state before the XMG aggressive pass.  Audit
discarded its first implementation because of an uninitialized variable
and an incomplete child-slot test.  The rewritten check matches the
collision shape; its place in the lane schedule remains a measured
Tier-2 choice.  The final $H_6$ logs show zero trigger/outcome
consistency violations, and s38417 reaches NDP 136{,}918 versus
137{,}173 with P6 disabled.  The registry therefore retains P6 under
its measured Tier-2 contract.

\subsection{Final-validation ledger}
\label{app:validation}

Table~\ref{tab:machine} links each final quantitative or finite claim to
the independent check that validates it.

\begin{table}[H]
\centering
\scriptsize
\setlength{\tabcolsep}{2.5pt}
\caption{Checks tied directly to final claims in this paper.}
\label{tab:machine}
\begin{tabular}{@{}p{0.22\columnwidth}p{0.17\columnwidth}p{0.35\columnwidth}p{0.13\columnwidth}@{}}
\toprule
Claim & Evidence type & Independent check & Result \\
\midrule
40$\to$31 exact cover & Source proof/audit & U-Q31 and its four
dependencies over the legal implementation domain & PASS \\
Stored 4-input NPN forms & Exhaustive decoder & Two independent
decoders; every driver canonicalizes to its claimed class & 222/222 \\
RS/RW/RF gain bounds & Source proof/audit & T8/T9 source obligations
and adversarial audits & PASS \\
Gain-bound regression & Finite log scan & No recorded violation in
\texttt{gain\_scan.json} & 66/66 \\
Gate-only trajectory & Differential output & Byte comparison after
normalizing the embedded timestamp & 66/66 identical \\
Optimizer correctness & CEC & Final untimed equivalence sweep over all
released outputs & 552/552 PASS \\
\bottomrule
\end{tabular}
\end{table}

The ledger separates three kinds of support.  Source audits discharge
universal proof obligations; finite decoders and scans discharge
enumerable obligations; differential and CEC runs validate produced
artifacts.  Each supports only the corresponding type of claim.

\newpage
\subsection{Representative refutations}
\label{app:refutations}

Table~\ref{tab:refuted} retains six consequential draft claims whose
counterexamples materially changed the final theory or implementation.

\begin{table}[H]
\centering
\scriptsize
\caption{Draft claims refuted during this campaign.}
\label{tab:refuted}
\begin{tabular}{@{}p{0.30\columnwidth}p{0.39\columnwidth}p{0.21\columnwidth}@{}}
\toprule
Draft claim & Counterexample or defect & Final disposition \\
\midrule
Cut4 cone size $\le7$ implies identity &
A two-record MAJ3 cone has stored size 1 and gain 1 &
Use E1--E5 accounting certificates \\
NPN-minimal cone implies cut4 identity &
External hash sharing lowers the unreferenced replacement count &
Retain the dynamic count check \\
First fiber witness has stored size 2 &
$\mathrm{maj}(\mathrm{maj}(x,y,z),x,y)$ absorbs to MAJ3 &
Replace by the verified three-gate polarity fiber \\
H08 $\equiv$ H10 &
Seven-AND witness: $(3,2)$ versus $(4,3)$ &
Keep both recipe actions \\
H27 dominated by H39 &
Nine-gate witness reverses the proposed relation &
Keep both recipe actions \\
Aggressive algebraic rewriting simulates selective &
Frozen-order family separates the two strategies &
Retain selective rewriting \\
\bottomrule
\end{tabular}
\end{table}

\paragraph*{Inventory citation rule}
Only claims marked current in the artifact support the results.
Withdrawn and refuted claims remain as provenance and do not serve as
premises.
 \fi

\end{document}